%% file: main.tex
\documentclass[]{article}

\usepackage[UKenglish]{babel}
\usepackage{algorithm}%
\usepackage{algorithmicx}%
\usepackage{algpseudocode}%
\usepackage{listings}
\usepackage{amsmath}
\usepackage{amssymb}
\usepackage{amsthm}
\usepackage{enumerate}
\usepackage{enumitem}
\usepackage{tikz}
\usepackage{tikz-cd}
\usetikzlibrary{arrows, arrows.meta, fit, positioning, shapes, backgrounds}
\usepackage[utf8]{inputenc}
\usepackage{textcomp}
\usepackage[stretch=10]{microtype}
\usepackage{xcolor}
\usepackage{bbm} 
\usepackage{makecell}
\usepackage{caption}
\usepackage{subcaption}
\usepackage{float}
\usepackage{multirow}
\usepackage[section]{placeins}
\usepackage{diagbox}
\usepackage{spverbatim}
\usepackage{todonotes}
\usepackage{pgfplots}
\usepackage{hyperref}
\hypersetup{                            
	colorlinks=true,                    
	linkcolor=black,                     
	citecolor=black,                     
	urlcolor=black                       
}
\usetikzlibrary{decorations.pathreplacing}
\usepackage{chngpage}
\usepackage[autostyle]{csquotes}
\usepackage{doi}
\usepackage[backend=biber, style=alphabetic, maxbibnames=99]{biblatex}
\graphicspath{{figures/}}

\makeatletter
\newtheorem*{rep@theorem}{\rep@title}
\newcommand{\newreptheorem}[2]{%
\newenvironment{rep#1}[1]{%
 \def\rep@title{#2 \ref{##1}}%
 \begin{rep@theorem}}%
 {\end{rep@theorem}}}
\makeatother

\theoremstyle{plain}
\newtheorem{thm}{Theorem}[section]
\newreptheorem{theorem}{Theorem}%

\newtheorem{lem}[thm]{Lemma}
\newtheorem{prop}[thm]{Proposition}
\newreptheorem{proposition}{Proposition}%

\newtheorem{cor}[thm]{Corollary}

\theoremstyle{definition}
\newtheorem{defn}[thm]{Definition}
\newtheorem{exmp}[thm]{Example}

\theoremstyle{remark}
\newtheorem{rem}[thm]{Remark}

\newcommand*\samethanks[1][\value{footnote}]{\footnotemark[#1]}

\newcommand{\N}{\mathbb{N}}

\newcommand{\R}{\mathbb{R}}

\newcommand{\eps}{\varepsilon}

\renewcommand{\pmod}{\mathbf{PersMod}}

\newcommand{\id}{\textnormal{id}}
\newcommand{\diam}{\textnormal{diam}}
\newcommand{\Dgm}{\textnormal{Dgm}}
\newcommand{\VR}{\textnormal{VR}}

\floatname{algorithm}{Procedure}

\pgfplotsset{compat=1.6}

\tikzstyle{vertex} = [circle, draw = black, fill = black, minimum size = 3mm]

\title{Bottleneck Profiles and Discrete Prokhorov Metrics for Persistence Diagrams}
\author{Pawe{\l} D{\l}otko\thanks{Mathematical Institute, Polish Academy of Sciences, Warsaw} \and Niklas Hellmer\samethanks
    }

\hypersetup{
pdftitle={Bottleneck Profiles and Discrete Prokhorov Metrics for Persistence Diagrams},
pdfsubject={cs.CG},
pdfauthor={Niklas ~Hellmer and Pawe\l{} ~D\l{}otko },
pdfkeywords={Persistence, Prokhorov metric, Topological Data Analysis},
}

\begin{document}
\setlength{\parindent}{0pt}
\maketitle

\begin{abstract}
    In topological data analysis (TDA), persistence diagrams have been a successful tool.
    To compare them, Wasserstein and Bottleneck distances are commonly used.
    We address the shortcomings of these metrics and show a way to investigate them in a systematic way by introducing bottleneck profiles.
    This leads to a notion of discrete Prokhorov metrics for persistence diagrams as a generalization of the Bottleneck distance.
    These metrics satisfy a stability result and can be used to bound Wasserstein metrics from above and from below.
    We provide algorithms to compute the newly introduced quantities and end with an discussion about experiments.
\end{abstract}


\section{Introduction}
The field of topological data analysis (TDA) is becoming a popular tool to study the structure of complex data.
One of the major tools of TDA is persistent homology (PH) \cite{EdelsbrunnerHarer2010CompTop}.
Its pipeline takes a (often highly complex) point cloud in Euclidean space as input and produces a point cloud in the plane, the persistence diagram (PD), as output.
Intuitively, persistence diagrams serve as a summary of the shape of the input data.
As a consequence, one can compare different shapes indirectly, by comparing their PDs.
The need for a robust and computationally efficient notion of distance for PDs arises.
Classically, one uses the Bottleneck and Wasserstein distances to this end \cite{Kerber2017Geometry}.
However, the Bottleneck distance only picks up the single biggest difference and the Wasserstein distance is prone to noise, as it picks up every difference no matter how small.

This fact motivates our work to search for new metrics.
Starting from the investigation of bottlenecks, we introduce the notion of the bottleneck profile of two PDs, which is a map $\R_{\ge 0} \to \N \cup\{\infty\}$ (Definition \ref{defn:DfunctionPD}).
This tool summarizes metric information at varying scales and generalizes the Bottleneck distance.
Also the Wasserstein distance can be, in special cases, computed from the bottleneck profile; in general, it can be bounded given a bottleneck profile.

The bottleneck profiles arises naturally in a discrete version of the Prokhorov distance, which is a classical tool in probability theory.
It turns out that the Bottleneck and the Prokhorov distance are just two instances of a whole family of Prokhorov-style metrics discussed in this paper (Definition \ref{defn:discreteProkhorov}).
This family is parameterised by subclass of functions $f\colon [0,\infty) \to [0,\infty)$.
Not every function $f$ gives in fact rise to a genuine metric; we examine the conditions on $f$ in which cases it does (Definition \ref{defn:AdmissibleFunction}, such $f$ are called \textit{admissible}).
\begin{reptheorem}{thm:ProkhorovIsMetric}
Fix an admissible function $f\colon \R_{\ge 0} \to \R_{\ge 0}$. The discrete $f$-Prokhorov metric is an extended pseudometric.
\end{reptheorem}
In addition to theoretical development, we discuss algorithms to compute the bottleneck profile and various Prokhorov-type distances.
\begin{repproposition}{prop:ProkhorovAlgo}
Let $f\colon [0,\infty)\to [0,\infty)$ be monotonically increasing. Assume that the values and preimages of $f$ can be computed in $O(1)$.
Then $\pi_f(X,Y)$ can be computed in $O(n^2 \log(n))$.
\end{repproposition}
We provide a run-time analysis and experiments on a number of data sets. The algorithms are provided as an open source implementation. 

\section{Background}
\subsection{Measure Theory}
Let $(X,d)$ be a metric space.
It is \textit{complete} if every Cauchy sequence has a limit in $X$.
It is \textit{separable} if it has a countable dense subset.
A complete separable metrizable topological space is called a \textit{Polish space}.
For example, all Euclidean spaces $\R^n$ are Polish.
Polish spaces are a convenient setting for measure or probability theory.

In general, we endow $X$ with the Borel $\Sigma$-algebra $\mathfrak{B}(M)$ and denote the set of probability measures by $\mathcal{P}(X)$.

Let us recall an important inequality \cite[p. 6]{grafakosClassicalFourierAnalysis2014}:
\begin{lem}[Chebychev's inequality]\label{lem:Chebychev}
    Let $(X,\Sigma, \mu)$ be a measure space and let $f\colon X \to \R$ be a measurable function.
    Then for any $p>0$ and $t>0$,
    \[
    \mu\{\vert f(x)\vert>t\} \le \frac{1}{t^p}\int\limits_{\vert f\vert\ge t}\vert f\vert^p d\mu.
    \]
\end{lem}

\subsubsection*{Metrics for Probability Measures}

There are various ways to compare different probability measures.

\begin{defn}\label{defn:MeasureWasserstein}
For $p\ge 1$, the $p$-\textit{Wasserstein metric} is
\begin{align*}
	W_p\colon \mathcal{P}(X) \times \mathcal{P}(X) & \to [0,\infty]\\
	W_p(\mu,\nu) & = \inf\limits_\gamma\left( \int d(x,y)^p d\gamma(x,y)\right)^\frac{1}{p},
\end{align*}
where the infimum is taken over all couplings $\gamma$ with marginals $\mu$ and $\nu$.
\end{defn} 
The $1$-Wasserstein metric is also known as \textit{Kantorovich metric} or \textit{earth mover's distance}.
The latter name is motivated by the idea of thinking about $\gamma$ as a transport plan for moving a pile of earth $\mu$ into the pile $\nu$.
The cost of transportation equals the distance by which the earth is moved.

Intuitively, there are two different ways to ``slightly change'' a measure.
The first one is to move all the mass by a tiny distance.
The second one is to move a tiny part of the mass arbitrarily, possibly very far away.
While the Wasserstein metric is stable under perturbations of the first kind, small changes of the second kind can result in large differences in the metric.
The Prokhorov metric \cite{prokhorovConvergenceRandomProcesses1956} seeks to resolve this problem.
It is constructed in such a way that an $\eps$-neighborhood of a measure is characterized as follows: 
One may move $\eps$ of the mass arbitrarily and the rest by at most $\eps$, see Figure \ref{fig:ProkhorovExmp} for an illustration.
\input{FigMeasureProkhorov}
We now formalize this idea.

For a Borel set $A\subset X$, the (open) $\eps$-ball around $A$ is
\[
    A_\eps = \{ x\in X \colon d(x,A)<\eps\}
\]

The \textit{Prokhorov metric} $\pi$ for two probability measures $\mu, \nu$ is defined as follows:
\begin{align*}
	\pi\colon\mathcal{P}(X) \times \mathcal{P}(X) & \to [0,\infty]\\
	\pi(\mu, \nu) &= \inf\{\eps>0 \colon \forall A \in \mathfrak{B}(X): \mu(A) \le \nu(A_{\eps})+\eps \text{ and}\\
	&\phantom{=inf\{} \nu(A) \le \mu(A_{\eps})+\eps\}.
\end{align*}
By Strassen's Theorem (cf. Remark 1.29 in \cite{villaniTopicsOptimalTransportation2003}, Appendix 1.4), an alternative  characterization of the Prokhorov metric is given in terms of couplings $\gamma$ which marginalize to $\mu$ and $\nu$ (compare Figure \ref{fig:CouplingExmp}),
\begin{align}
	\pi(\mu, \nu) &= \inf \{ \inf\{\eps>0 \colon\gamma(\{(x,y)\colon d(x,y)>\eps\})<\eps\} \colon\notag \\ & \phantom{= inf \{} \gamma \text{ has marginals } \mu \text{ and  } \nu\}.\label{eqn:MeasureProkhorov}
\end{align}
This allows for a discretization suitable for persistence diagrams, see Section \ref{sec:PDMetrics}.

\begin{exmp}
    Let $x_1, x_2 \in X$ with $d(x_1,x_2)<1$ and consider the Dirac measures $\delta_{x_1}, \delta_{x_2}$.
    We claim that $\pi(\delta_{x_1}, \delta_{x_2}) = d(x_1,x_2)$.
    The only coupling with correct marginals is $\delta_{(x_1,x_2)}$.
    Then we have
    \[
        \delta_{(x_1,x_2)}(\{(x,y)\colon d(x,y)>\eps\}) = \begin{cases} 1 & \text{if } \eps<d(x_1,x_2),\\
        0 & \text{otherwise;}
        \end{cases}
    \]
    we write $\mathbbm{1}_{\eps<d(x_1,x_2)}$ as a shorthand notation for the right hand side.
    Consequently, we have
    \[
        \inf\{\eps >0 \colon \mathbbm{1}_{\eps<d(x_1,x_2)}<\eps\} = d(x_1, x_2).
    \]
\end{exmp}

\input{FigCoupling}

In their survey \cite{gibbs}, Gibbs and Su show that the $1$-Wasserstein can be related to the Prokhorov metric via
\[
	\pi^2 \le W_1 \le (1+ \text{diam}(X)) \; \pi,
\]
where $\text{diam}(X)$ is the diameter of the underlying space. 
We provide discrete analogues of this estimate in Propositions \ref{prop:ProkhorovLeWasserstein} and \ref{prop:WassersteinLeProkhorov}.

For more on metrics of probability measures, see the book \cite{rachevProbabilityMetricsStability1991}; references for optimal transport include \cite{peyreComputationalOptimalTransport2020} which takes on a computational perspective.

\subsection{Persistent Homology}
\begin{defn}
The category $\pmod$ is the functor category 
\[
\pmod = \text{Fun}(\R, k\textbf{-mod}^{\text{fd}})
\]
from the reals as a poset category to finite dimensional vector spaces.
Its objects are called \textit{pointwise finite dimensional (p.f.d.) persistence modules}.
A p.f.d. persistence module $A = (A_t)_{t\in \R}$ comes with \textit{transition maps}
\[
    A_{s\leq t} \colon A_s \to A_t
\]
for $s\le t$.
\end{defn}

\begin{defn}
An \textit{interval module} for an interval $J\subset \R$ is a p.f.d persistence module with 
\[
    \mathbb{I}_{J, t} = \begin{cases}k &\textnormal{if }t\in J,\\
    0&\textnormal{otherwise,}
    \end{cases}
\]
and
\[
    \mathbb{I}_{J, s}\to \mathbb{I}_{J, t} = \begin{cases}\id &\textnormal{if } s,t\in J,\\
    0&\textnormal{otherwise,}
    \end{cases}
\]
for $s\leq t$.
The start and endpoint of $J$ are referred to as \textit{birth time} $b(J)$ and \textit{death time} $d(J)$, respectively.
Their difference $d(J)-b(J)$ is called the \textit{persistence} or \textit{lifetime} of an interval.
\end{defn}
Note that we do not specify whether the endpoints are contained in the interval; they may be $\pm \infty$.
Interval modules are of special interest because p.f.d. persistence modules admit an interval decomposition.
\begin{thm}[\cite{Crawley-Boevey2015Decomposition}, Theorem 1.1]
Let $\mathbb{A}\in \pmod$. Then there exists a collection of intervals $\mathcal{J}$ such that
\[
    A\cong \bigoplus\limits_{J\in \mathcal{J}}\mathbb{I}_J.
\]
\end{thm}

Such an interval decomposition (sometimes called \textit{barcode)} can be visualized via a persistence diagram.
\begin{defn}
A \textit{persistence diagram} (PD) is multiset of points in $\R^2$, consisting of
\begin{itemize}
    \item points above the diagonal $(b,d), b<d$, each with finite multiplicity and
    \item each point on the diagonal $\Delta = \{(s,s)\in \R^2\}$ with countable multiplicity.
\end{itemize}
\end{defn}
The convention to include diagonal points with infinite multiplicity will be useful for the construction of distances between persistence diagrams.

To obtain a PD from the above interval decomposition, collect the birth and death times of the intervals
\[
    \Dgm(A) = \{(b,d) \in (\R\cup \{\infty\})^2 \colon \langle b,d\rangle \in \mathcal{J}\}
\]
(where the angled brackets indicate that the endpoints may or may not be included);
add all the points on the diagonal with countable multiplicities.
Off-diagonal points have finite multiplicities since the persistence module is pointwise finite dimensional.
We will freely identify off-diagonal points in the diagram with the corresponding interval.
Points close to the diagonal have a short lifetime and are often regarded as noise.
To compare persistence diagrams, we consider one-to-one correspondences between them.
To take care of different cardinalities of off-diagonal points and to get rid of noisy, short-lifetime points, we allow them to be mapped to the diagonal. 
This explains the inclusion of the diagonal with infinite multiplicity in the above definition.

\begin{defn}
A \textit{matching} $\eta$ between persistence diagrams $X$ and $Y$ is a bijection which fixes all but finitely many diagonal points. 
The \textit{cardinality} or \textit{size} of  a matching $\eta$, denoted by $\left\vert\eta\right\vert$ is the number of points which are not fixed.
\end{defn}

\begin{defn}\label{def:BottleneckDistance}
The \textit{bottleneck distance} between two persistence diagrams $X,Y$ is
\[
    W_\infty(X,Y) = \inf \limits_{\eta\colon X\to Y} \sup\{d(x, \eta(X))\colon x\in X\}, 
\]
where $\eta$ ranges over all matchings.
\end{defn}
\begin{defn}\label{def:WassersteinDistance}
Let $1\le p <\infty$.
The $p$-\textit{Wasserstein distance} between two persistence diagrams $X,Y$ is
\[
    W_p(X,Y) = \inf\limits_{\eta\colon X \to Y} \left(\sum \limits_{x \in X} d(x, \eta(x))^p\right)^\frac{1}{p},
\]
where $\eta$ ranges over all matchings.
\end{defn}

The notation of Definition \ref{def:WassersteinDistance} has the advantage of being compact, but note that we have uncountably many summands. Although usually, only finitely many, namely $\vert\eta\vert$ for the optimal matching $\eta$, will be non-zero.
Similarly, also only finitely many elements of the uncountable set of which we take the supremum in Definition \ref{def:BottleneckDistance} are non-zero.
\begin{defn}
Let $p\geq 1$.
We say a persistence diagram $X$ has \textit{finite pth moment}, if the $p$-Wasserstein distance to the empty diagram is finite: $W_p(X,\emptyset) <\infty.$
\end{defn}
Except from section \ref{subsec:MetricProperties}, the persistence diagrams in this paper are assumed to have finitely many off-diagonal points.
Therefore, the infima in definitions \ref{def:BottleneckDistance} and \ref{def:WassersteinDistance} are actually minima.

Notice the analogy between Definitions \ref{defn:MeasureWasserstein} and \ref{def:WassersteinDistance}.
We replace probability measures by counting measures and hence turn the integral into a sum.
The infimum is taken over all matchings instead of all couplings.
This observation will serve as a blueprint for the construction of the discrete Prokhorov metric for persistence diagrams in the Section \ref{sec:PDMetrics}.

The motivation to compare persistence diagrams comes from topological data analysis, where they serve as a sumray statistic of topological information.
\begin{exmp}
Given a finite subset $X$ of some metric space, we can consider the Vietoris-Rips complex, cf. \cite{EdelsbrunnerHarer2010CompTop} III.2.
This is a filtered simplicial complex; for filtration value $r>0$ it is given by
\[
    \VR(X)[r] = \{S \subset X \colon \diam(S)<2r \}.
\]
Applying the homology functor gives rise to a persistence module:
For $r\le s$, the inclusion
\[
    \VR(X)[r] \hookrightarrow \VR(X)[s]
\]
induces a module homomorphism
\[
    H_\ast(\VR(X)[r]) \to H_\ast(\VR(X)[s]).
\]
This is called the \textit{(Vietoris-Rips) persistent homology} of $X$, denoted $PH_\ast(X)_t$.
Summands in its interval decomposition are interpreted as topological features which are ``born'' at a certain point in the filtration and ``persist'' for some time.
They are regarded to be more significant the longer the intervals are. 
The following theorem ascertains that this is a useful tool.
\begin{thm}[{{\cite[Theorem 3.1]{chazal2009gromov}}}]\label{thm:ChazalStability}
Let $X,Y$ be finite metric spaces, fix some $k\ge 0$. 
Then we have
    \[W_\infty(\Dgm(PH_k(X)), \Dgm(PH_k(Y))) \le 2d_{\textnormal{GH}}(X,Y),\]
    where $d_{\text{GH}}$ is the Gromov-Hausdorff distance \cite{gromov2007metric}.
    \footnote{Note that \cite{chazal2009gromov} uses a different notion of Gromov-Hausdorff distance, which is equal to two times $d_{GH}$ of \cite{gromov2007metric}.}
\end{thm}
In other words, if we change the input point cloud by $\eps$ in the Gromov-Hausdorff metric, the resulting PDs differ by at most $\eps$ in the Bottleneck distance.
\end{exmp}

\section{Bottleneck Profiles}\label{sec:BottleneckProfiles}
\begin{figure}[!tphb]
    \centering
    \includegraphics[width=\textwidth]{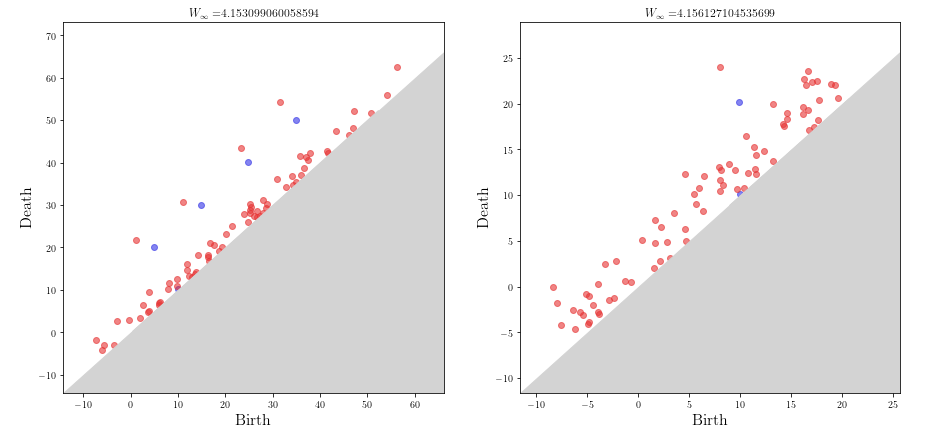}
    \caption{Four bottlenecks on the left, a single bottleneck on the right, realizing almost the same bottleneck distance.}
    \label{fig:1and4Bottlenecks}
\end{figure}
The bottleneck distance $W_\infty$ has a major drawback: 
It only captures the single most extreme difference between two persistence diagrams.
This implies that the same bottleneck distance can be realized by different pairs of persistence diagrams, cf Figure \ref{fig:1and4Bottlenecks}.
We introduce the notion of the bottleneck profile to address the topic of secondary, tertiary,... bottlenecks and their multiplicities.

\begin{defn}\label{defn:DfunctionPD} Given two persistence diagrams $X, Y$, define their \textit{bottleneck profile} to be
\[
    D_{X,Y}\colon \R_{\ge 0} \to \N \cup \{\infty\},\qquad t \mapsto \inf\limits_{\eta\colon X \to Y} \lvert\{x \colon d(x,\eta(x))>t\}\rvert;
\]
where $\vert \cdot \vert$ denotes the cardinality of the set.
\end{defn}
For $d: \R^2 \times \R^2 \to \R_{\ge 0}$ we take an $\ell^p$-metric $d(x,y) = \|x-y\|_p$, where the choice of $p$ might depend on the setting. 
For example, when comparing with the $p$-Wasserstein distance, one might like to choose this same $p$.

Since the infimum is taken over a subset of the natural numbers, it is actually a minimum.
To be consistent with the notation in definitions \ref{def:BottleneckDistance} and \ref{def:WassersteinDistance}, we choose to adhere to the use of infimum.

The following observation is immediate:

\begin{lem}
The bottleneck profile $D_{X,Y}$ is monotonically decreasing.
\end{lem}
\begin{proof}
Let $\eta\colon X \to Y$ be any matching realizing $D_{X,Y}(s)$ for some $s$.
Let now $t>s$, then every distance longer than $t$ is in particular longer than $s$ and consequently
\[
    \lvert\{x \colon d(x,\eta(x))>t\}\rvert \le \lvert\{x \colon d(x,\eta(x))>s\}\rvert = D_{X,Y}(s).
\]
Taking the infimum over all matchings decreases the left hand side and yields $D_{X,Y}(t)$.
\end{proof}

Knowing this, it is interesting when the bottleneck profile becomes zero.
\begin{lem}\label{lem:DBottleneck}
$D_{X,Y}(t) = 0 \Leftrightarrow t \ge W_\infty(X,Y)$.
\end{lem}
\begin{proof}
By definition, the bottleneck distance is the smallest $t>0$ such that there is a matching mapping all points within distance $t$.
In formulas,
\begin{align*}
W_\infty(X,Y) &= \inf\{t>0 \colon \inf\limits_{\eta\colon X \to Y} \lvert\{x \colon d(x,\eta(x))>t\}\rvert = 0\}\\
&=\inf\{t>0 \colon D_{X,Y}(t) = 0\}.
\end{align*}
\end{proof}

Thus we recover the bottleneck distance from the bottleneck profile.
The bottleneck cost of a matching is the longest distance over which two points are matched.
Minimizing the bottleneck cost over all matchings yields the bottleneck distance, which we can think of as the primary bottleneck.
Similarly, the secondary bottleneck cost of a matching is the second longest distance over which two points are matched.
Taking the minimum over all matchings here gives a notion of a secondary bottleneck, which equals $\inf\{t>0 \colon D_{X,Y}(t) \leq 1\}$ by an argument analogous to the previous proof.
This motivates the name bottleneck profile.

\begin{exmp}\label{exmp:singletons}
Let $X=\{x\}$ and $Y=\{y\}$ both consist of one point each and assume that $d(x,y)<d(x,x')+d(y,y')$, where the prime denotes the projection to the diagonal.
That means that $x\mapsto y$ is an optimal matching.
Consequently, the bottleneck profile looks as follows:
\[
    D_{X,Y}(t) = \begin{cases}
        1 & \text{if } 0\le t \le d(x,y),\\
        0 & \text{if } t> d(x,y).
    \end{cases}
\]

\end{exmp}

\begin{exmp}\label{exmp:StableRank}
If we take one of the persistence diagrams to be the empty one, there is only one choice of matching: everytthing is paired with the diagonal.
As a consequence,
\[
D_{X,\emptyset}(t) = \lvert\{x = (x_1, x_2)\colon \frac{x_2-x_1}{2}>t\}\rvert = \lvert\{x = (x_1,x_2) \colon x_1 + 2t < x_2\}\rvert.
\]
This is also known as the stable rank function corresponding to the contour $C(a, \eps) = a + 2 \eps$, introduced in \cite{chacholskiMetricsStabilizationOne2020},
which counts the bars of $X$ of length $>2t$.
\end{exmp}

\begin{figure}[tbh]
    \centering
    \makebox[\textwidth][c]{
    \includegraphics[width=1.5\textwidth]{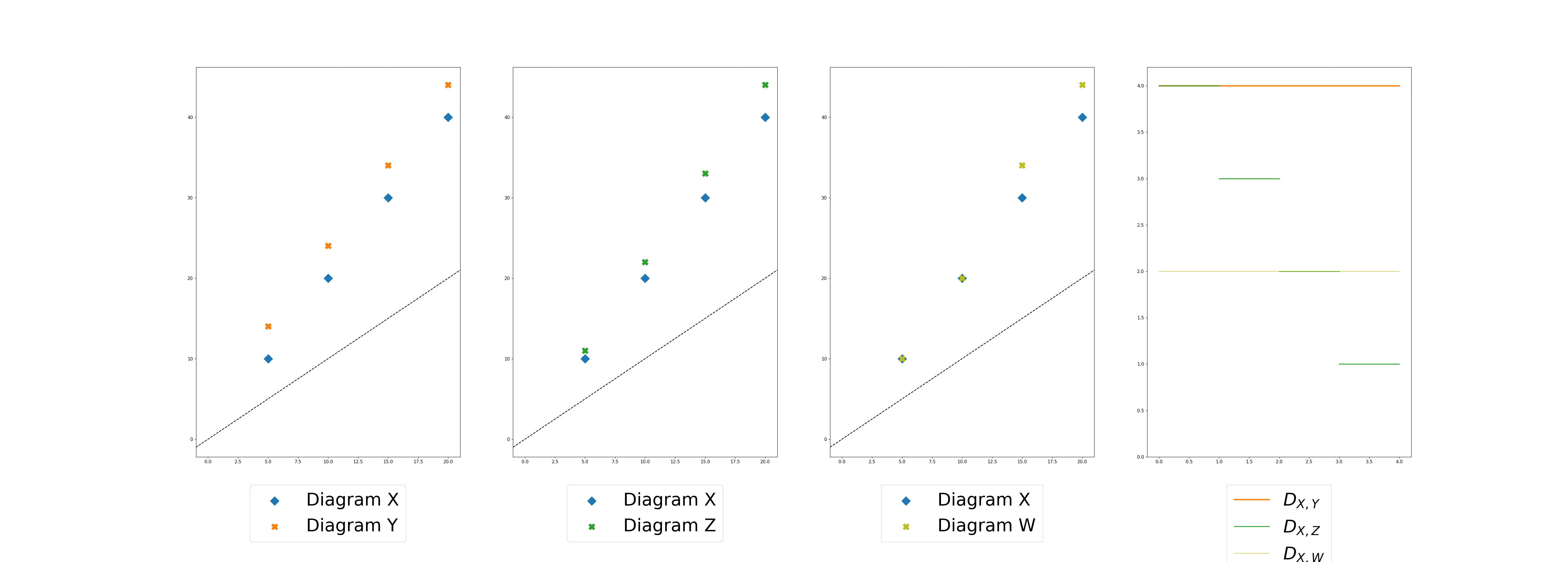}}
    \caption{\label{fig:D_linear}The PD $X$ has bottleneck distance $3$ to each of the PDs $Y,Z, W$ (first three images).
    However, it is attained with different multiplicities, which one can read off from the bottleneck profile (right-most image)}
    
\end{figure}
\begin{exmp}\label{exmp:bottleneckprofile}
Consider some particular simple persistence diagrams.
The first three parts of Figure \ref{fig:D_linear} each show a base diagram (``Diagram $X$'', in blue) with four points and perturbations of it:
The orange diagram (``Diagram $Y$'') in the first image is obtained by shifting the blue one by three.
The green diagram (``Diagram $Z$'') shifts the top point of $X$ by three, the next point by two, the third by one and leaves the lowest point unchanged.
For the yellow diagram (``Diagram $W$'') in the third image, we only shift two points from $X$ by three and leave the other two untouched.
Clearly, the bottleneck distance between the base diagram and each of the shifted versions is three.
But the amount of shifted points is reflected in the bottleneck profile:
While $D_{X,Y}(t)$ is four, $D_{X,Z}(t)$ is two (i. e. the multiplicity of the bottleneck) for $0<t<3$.
And $D_{X,W}$ displays more steps, reflecting the fact that there are secondary and tertiary bottlenecks.
\end{exmp}

Note that the function $D$ enjoys some properties reminiscent of a metric (hence the notation $D$):
It is obviously symmetric.
The triangle inequality does not hold pointwise but in a scaled version, that is:
\begin{lem}\label{lem:TriangleD}
For all persistence diagrams $X,Y,Z$ and all real numbers $s,t\ge 0$,
$D_{X,Z}(s+t)\le D_{X,Y}(s)+D_{Y,Z}(t)$.
\end{lem}
\begin{proof}
This follows from the triangle inequality on $\R^2$.
Fix $s,t\ge 0$, let $\eta_{X,Y} \colon X \to Y$ and $\eta_{Y,Z}\colon Y \to Z$ denote optimal matchings realizing $D_{X,Y}(s)$ and $D_{Y,Z}(t)$, respectively.
Let $\eta = \eta_{Y,Z}\circ \eta_{X,Y}\colon X \to Z$ be the matching obtained by composition.
It suffices to show that
\[
    \lvert\{x\colon d(x,\eta(x))>s+t\}\rvert\le \lvert\{x\colon d(x,\eta_{X,Y}(x))>s\}\rvert + \lvert\{y\colon d(y,\eta_{Y,Z}(y))>t\}\rvert,
\]
because the left hand side only decreases if we take the infimum over all matchings.
Hence we have to investigate what happens when a point $x$ is matched to $\eta(x)$ which is farther apart than $s+t$. 
Note that $\eta(x) = \eta_{Y,Z}( \eta_{X,Y}(x))$, so we compare the distances of the matched points using the triangle inequality,
\[
    s+t < d(x,\eta(x)) \le d(x,\eta_{X,Y}(x)) + d(\eta_{X,Y}(x), \eta(x)).
\]
\begin{figure}[!htbp]
    \centering
    \begin{tikzpicture}
            \node[vertex] (v1) {};
            \node[vertex] (v2) [right=of v1] {};
            \node[vertex] (v3) [below=of v1] {};
            
            \draw[thick] (v1)--(v2)--(v3)--(v1);
            
            \node[left] at (v1.west){$\eta_{X,Y}(x)$};
            \node[right] at (v2.east){$\eta(x) = \eta_{Y,Z}( \eta_{X,Y}(x))$};
            \node[below] at (v3.south){$x$};
        \end{tikzpicture}
    \caption{The situation in the proof of Lemma \ref{lem:TriangleD}}
    \label{fig:Triangle}
\end{figure}
Therefore, it cannot be that both $d(x,\eta_{X,Y}(x))\leq s$ and $d(\eta_{X,Y}(x), \eta(x))\leq t$ (compare Figure \ref{fig:Triangle}).
That means, we have $d(x,\eta_{X,Y}(x))> s$ or $d(\eta_{X,Y}(x), \eta(x))> t$ or both.
Using the principle of inclusion-exclusion, conclude
\begin{align*}
    \lvert\{x\colon d(x,\eta(x))>s+t\}\rvert &= \lvert\{x\colon d(x,\eta_{X,Y}(x))>s\}\rvert + \lvert\{y\colon d(y,\eta_{Y,Z}(y))>t\}\rvert\\
    &\phantom{=}- \lvert\{x \in \colon d(x,\eta_{X,Y}(x))>s \text{ and } d(\eta_{X,Y}(x),\eta_{Y,Z}(\eta_{X,Y}(x))) > t\}\rvert\\
    &\leq \lvert\{x\colon d(x,\eta_{X,Y}(x))>s\}\rvert + \lvert\{y\colon d(y,\eta_{Y,Z}(y))>t\}\rvert.
\end{align*}
\end{proof}

Note that $D_{X,Y}(t)=0$ for all $t>0$ implies $X=Y$ only under some finiteness assumptions.
For example, consider a converging sequence $(a_n)_{n\in \N} \subset \R^2$ above the diagonal with limit $a\not\in (a_n)$, which is also above the diagonal.
Set $X$ to consist of all elements of the sequence $\{a_n \colon n\in \N\}$.
Set $Y$ to be $X \cup \{a\}$.
Then for all $\eps > 0$ there exists $\eta \colon X\to Y$ such that $d(x, \eta(x))<\eps$ for all $x\in X$.
Therefore, $D_{X,Y}(t) = 0$ for every $t>0$, but $X\neq Y$.

Following \cite{blumbergRobustStatisticsHypothesis2014}, we denote by $\Bar{\mathcal{B}}$ the set of persistence diagrams such that for each $\eps>0$ there are finitely many points of persistence $>\eps$.
The next lemma is an immediate consequence of \cite[Lemma 3.4]{blumbergRobustStatisticsHypothesis2014}.
\begin{lem}\label{lem:D=0impliesX=Y}
The bottleneck profile satisfies $D_{X, X}(t) = 0$ for all Persistence diagrams $X$ and $t>0$.
Moreover, $D_{X, Y}(t) = 0$ for all $t>0$ implies $X=Y$ for $X, Y\in \bar{\mathcal{B}}$.
\end{lem}
\begin{proof}
If $D_{X, Y}(t) = 0$ for all $t>0$, then $W_\infty(X,Y)=0$ by Lemma \ref{lem:DBottleneck}.
Now for $X,Y \in \bar{\mathcal{B}}$, this only happens if $X=Y$ by \cite[Lemma 3.4]{blumbergRobustStatisticsHypothesis2014}. 
\end{proof}

\subsection{Relation to  Wasserstein distances}
We have already seen how the bottleneck profile is related to the bottleneck distance. 
This is actually part of a more general result comparing it to $p$-Wasserstein metrics.

\begin{lem}\label{lem:DWasserstein}
Let $X, Y$ be two persistence diagrams, and let $p>0$. Then
\[
    D_{X,Y}(t)  \leq \frac{1}{t^p}W_p(X,Y)^p.
\]
\end{lem}
\begin{proof}
This follows from the Chebychev inequality (Lemma \ref{lem:Chebychev}) for counting measures.
To spell out the details, estimate that for every bijection $\eta$
\begin{align*}
    \lvert\{x \colon d(x,\eta(x)) >t\}\rvert 
        &= \sum \limits_{\{x \colon d(x,\eta(x)) >t\}} 1\\
        &\leq \sum \limits_{\{x \colon d(x,\eta(x)) >t\}} \frac{d(x,\eta(x))^p}{t^p}\\
        &\leq \sum \limits_{x\in X} \frac{d(x,\eta(x))^p}{t^p}\\
        &=\frac{1}{t^p}\sum \limits_{x\in X} d(x,\eta(x))^p.
\end{align*}
Now choosing $\eta$ to minimize the right hand side, we have by definition of the Wasserstein distance an estimate for $D_{X,Y}$:
\[
    D_{X,Y}(t) \leq \lvert\{x \colon d(x,\eta(x)) >t\}\rvert \leq \frac{1}{t^p}W_p(X,Y)^p.
\]
\end{proof}
This is illustrated by Figure \ref{fig:D+Wasserstein}.
\begin{figure}[htb]
    \centering
    \includegraphics[width=0.7\textwidth]{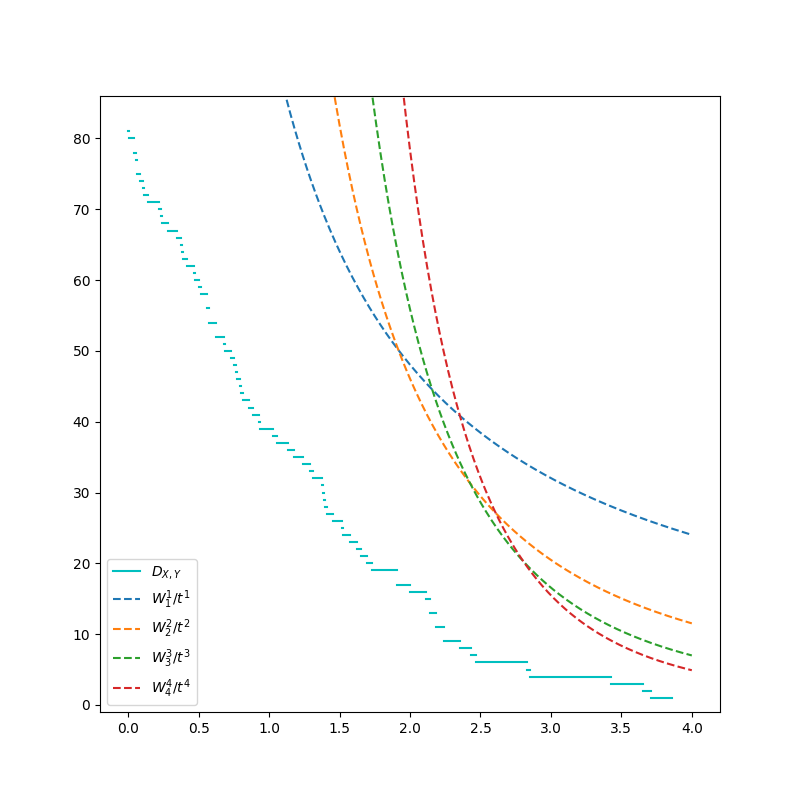}
    \caption{An example for the relation between $D_{X,Y}$ and the Wasserstein distance. 
    }
    \label{fig:D+Wasserstein}
\end{figure}
Note that we recover Lemma \ref{lem:DBottleneck} in the limit for $p \to \infty$:
\[
\left(\frac{W_p}{t}\right)^p \longrightarrow \begin{cases}
\infty & \text{if }t<W_\infty\\
1 & \text{if }t=W_\infty\\
0 & \text{if }t>W_\infty.
\end{cases}
\]
For $1$-Wasserstein, we have a further estimate:

\begin{lem}\label{lem:IntBPleqW1}
$ \int_0^\infty D_{X,Y}(t) dt \leq W_1(X,Y) $.
\end{lem}
\begin{figure}
    \centering
    \begin{tikzpicture}
        \draw[->] (0,0.5)--(0,5);
        \draw[->] (0,0.5)--(6,0.5);
        \node[left] at (0,3) {\footnotesize$i$};
        \node[left] at (0,2) {\footnotesize$i-1$};

        \draw[draw=gray,dashed,fill=gray] (0,2) rectangle ++(3,1);
        \draw [
        thick,
        decoration={
        brace,
        mirror,
        raise=0.05cm
        },
        decorate
        ] (0,2) -- (3,2);
            
        \draw [blue,[-)] (0.5,3)--(3,3);
        \draw[blue, [-)] (3,2)--(5,2);
        \node[blue, right] at (5,2) {\footnotesize$\lvert\{x \colon d(x,\eta(x))>t\}\rvert$};
        
        \draw (3,.4)--(3,.6);
        \draw[dashed] (3,.4)--(3,2);
        \node[below] at (3,0.5) {\footnotesize$\inf\{t>0\colon\lvert\{x\colon d(x,\eta(x))>t\}\rvert<i\}$};
    \end{tikzpicture}
    \caption{Illustrating the proof of Lemma \ref{lem:IntBPleqW1}: Decomposing the area under the graph into rectangles.}
    \label{fig:RectanglesUnderCurve}
\end{figure}
\begin{proof}
Let $\eta\colon X\to Y$ be the matching realizing $W_1(X,Y)$.
We compute the area under the graph of the function $t\mapsto \lvert\{x \colon d(x,\eta(x))>t\}\rvert$, which is piece-wise constant.
Decomposing it into rectangles of height one yields a width of $\inf\{t>0\colon\lvert\{x\colon d(x,\eta(x))>t\}\rvert<i\}$ for $i\geq 1$, cf. Figure \ref{fig:RectanglesUnderCurve}.
The width of the $i$th rectangle is the length of the $i$th longest edge in the matching.
Summing over all $i$ is therefore the same as summing the distances over which points are matched.
 In formulas:
\begin{align*}
    W_1(X,Y) & = \sum\limits_{x\in X} d(x,\eta(x))\\
     &= \sum \limits_{i \ge 1} \inf\{t> 0 \colon \lvert\{x \colon d(x,\eta(x))>t\}\rvert < i\}\\
    & = \int _0 ^\infty \lvert\{x \colon d(x,\eta(x))>t\}\rvert dt\\
    & \ge \int_0^\infty D_{X,Y}(t) dt.
\end{align*}
\end{proof}
\begin{prop}
If the bottleneck profile $D_{X,Y}(t)$ can be realized by the same matching $\eta$ for all $t>0$, then $\eta$ realizes $W_1(X,Y)$.
\end{prop}
\begin{proof}
If $\eta$ realizes $D_{X,Y}(t)$ for all $t>0$, then the inequality in the proof of the previous lemma becomes an equality
\[
    \int_0^\infty D_{X,Y}(t) dt = \int _0 ^\infty \lvert\{x \colon d(x,\eta(x))>t\}\rvert dt = \sum\limits_{x\in X} d(x,\eta(x)) \overset{(\ast)}{\geq} W_1(X,Y).
\]
Combining this with Lemma \ref{lem:IntBPleqW1}, we obtain
\[
    \int_0^\infty D_{X,Y}(t) dt = W_1(X,Y).
\]
Consequently, the inequality $(\ast)$ is actually an equality, which is what we wanted to prove.
\end{proof}

\subsection{Algorithms}
Recall the definition
\begin{align*}
    D_{X,Y}(t) &= \inf\limits_{\eta} \lvert\{x \colon d(x,\eta(x))>t\}\rvert,
\end{align*}
and let $\eta$ be the matching realizing the infimum.
Then $\eta$ also realizes the following supremum:
\[
    \sup\limits_{\eta} \lvert\{x \colon d(x,\eta(x))\le t\}\rvert,
\]
and consequently
\[
    D_{X,Y}(t) = \lvert\eta\rvert - \sup\limits_{\eta} \lvert\{x \colon d(x,\eta(x))\le t\}\rvert.
\]
Here, $\lvert\eta\rvert$ denotes the number of matched pairs which involve at least one off-diagonal point.
The computation of $\sup\limits_{\eta} \lvert\{x \colon d(x,\eta(x))\le t\}\rvert$ is a version of the unweighted maximum cardinality bipartite matching problem.
First, set up the following notation (following \cite[chapter VIII.4]{EdelsbrunnerHarer2010CompTop}).
Denote by $X_0$ the off-diagonal points of $X$ and by $X_0'$ their projections to the diagonal (and analogously for $Y$).w
Set $U = X_0 \cup Y_0'$ and $V = Y_0 \cup X_0'$ and consider the bipartite graph $G=(U \cup V, E)$ with $e=\{u,v\} \in E$ if either of the following holds:
\begin{itemize}
    \item $u\in X_0, v\in Y_0$ and $d(u,v) \le t$,
    \item $u\in X_0, v \in X_0'$ is its projection to the diagonal and $d(u,v) \le t$,
    \item $v\in Y_0, u \in Y_0'$ is its projection to the diagonal and $d(u,v) \le t$,
    \item $u \in Y_0'$ and $v \in X_0'$.
\end{itemize}

Let $M\subset E$ be a matching of maximal cardinality.
Observe that such a matching corresponds to a bijection $\eta \colon X \to Y$ maximizing $\lvert\{x \colon d(x,\eta(x))\le t\}\rvert$.

To estimate the run-time of this algorithm, let $n = \vert X\vert+\vert Y\vert$.
We solve the unweighted maximum cardinality bipartite matching problem using the Hopcroft-Karp algorithm \cite{hopcroft_n52_1973}.
Let us briefly recall this classical algorithm.
The algorithm extends a partial matching $M$ until it reaches a maximum one.
It achieves this by augmenting paths: A path $p$ that starts at an unmatched vertex in $U$ and ending at an unmatched vertex in $V$ such that edges from $U$ to $V$ are not in $M$ but edges from $V$ to $U$ are.
Removing edges from $p\cap M$ from the matching and instead inserting edges from $p\cap(E\setminus M)$ increases the size of $M$ by one.
The Hopcroft-Karp algorithm finds vertex-disjoint augmenting paths in $O(n^2)$ via the so-called \textit{layer subgraph}, which is constructed via a depth-first search in $O(n^2)$.
After extending the matching using all these augmenting paths, the algorithm starts over.
The algorithm terminates after $O(\sqrt{n})$ of these iterations.

While this consequently takes $O(n^{2.5})$ in the worst case, we perform a variant which exploits the geometric nature of the setting, as suggested in \cite{Efrat2001Geometry}.
Instead of building the layer graph explicitly, one can use a geometric data structure that allows for querying neighbors within a given distance, as well as removing points.
Following \cite{Kerber2017Geometry}, k-d trees achieve this requiring $O(\sqrt{n})$ for either of the two operations.
Consequently, as noted by \cite{Kerber2017Geometry} and \cite{Efrat2001Geometry}, our variant of the Hopcroft-Karp algorithm runs in $O(n^2)$.
Summarizing, we find the following:
\begin{prop}\label{prop:DRuntime}
Let $X, Y$ be finite persistence diagrams and denote $n = \vert X\vert +\vert Y\vert $.
The value of the bottleneck profile at $t$, $D_{X,Y}(t)$, can be computed in $O(n^2)$.
\end{prop}

\begin{rem}
Using k-d trees is useful in practice, but does not yield optimal theoretical run-times.
Indeed, the more sophisticated data structure from \cite{Efrat2001Geometry}, Section 5.1, can be constructed in $O(n \log(n))$.
The two relevant operations on it require $O(\log(n))$, so that the bottleneck profile could be evaluated in $O(n^{1.5}\log(n))$ using this method.
\end{rem}

\begin{rem}
Instead of using Hopcroft-Karp, one can regard the matching problem as a linear program.
For each $x \in X$ and $y\in Y$, we have a binary variable $f_{xy}$ indicating whether the edge from $x$ to $y$ is in the matching.
The coefficients (the cost of the edge) are given by 
\[
    c_{xy} = \begin{cases}
    1 &\text{if } d(x,y) > t,\\
    0 & \text{otherwise.}
    \end{cases}
\]
The objective is 
\begin{align*}
\text{minimize } &\sum \limits_{x,y} c_{xy}f_{xy}\\
\text{subject to }    &\forall x \in X\colon \sum\limits_y f_{xy} = 1, \;
    \forall y \in Y\colon \sum\limits_x f_{xy} = 1.
\end{align*}

\end{rem}

\section{Discrete Prokhorov Metrics for Persistence Diagrams}\label{sec:PDMetrics}
A straight-forward discretization of the coupling characterization of the probabilistic Prokhorov metric \eqref{eqn:MeasureProkhorov} gives the main notion of this section.
\begin{defn}\label{defn:discreteProkhorov}
Given two persistence diagrams $X, Y$, consider matchings $\eta\colon X \to Y$ to define their \textit{Prokhorov distance} as
\begin{align*}
    \pi(X,Y) &= \inf\{t>0\colon D_{X,Y}(t)<t\}\\
    &= \inf\{t>0\colon \inf\limits_{\eta\colon X\to Y} \lvert\{x \colon d(x,\eta(x))>t\}\rvert<t\}.
\end{align*}
\end{defn}

Informally, we look at the intersection of the bottleneck profile with the diagonal.
Similarly, we have already seen that the bottleneck distance arises as the intersection of $D_{X,Y}$ with the horizontal axis.
This motivates the the question, what functions we can intersect the bottleneck profile with to obtain a sensible notion of distance.

\begin{defn}\label{defn:AdmissibleFunction}
Consider a function $f\colon [0,\infty] \to [0,\infty]$.
We say $f$ is \textit{superadditive} if for any $s,t \geq 0$ we have $f(s+t) \geq f(s) + f(t)$.
A superadditive function $f$ is called \textit{admissible} if $\lim\limits{t\searrow0}f(t) = 0$.
Furthermore, the function $f\equiv 1$ is also said to be admissible.
\end{defn}
Notice that such superadditive functions are monotonically non-decreasing.
For example, any linear function with non-negative slope is admissible.
Moreover, increasing convex functions $f$ with $f(0)=0$ are admissible
For instance polynomials with non-negative coefficients and absolute term zero fulfill this criterion.

\begin{defn}
Given a fixed admissible function $f\colon \R_{\ge 0} \to \R_{\ge 0}$, define for any two PDs $X,Y$ their $f$\textit{-Prokhorov distance} to be
\begin{align*}
    \pi_f(X,Y) &= \inf\{t>0\colon D_{X,Y}(t)<f(t)\}\\
    &= \inf\{t>0\colon \inf\limits_{\eta} \lvert\{x\colon d(x,\eta(x))>t\}\rvert<f(t)\} .
\end{align*}

\end{defn}
Plugging in $f = id$ gives the Prokhorov distance, plugging in $f\equiv 1$ recovers the bottleneck distance (this is why this function is admissible even though it is not superadditive).

Intuitively, for $n\in \N$, plugging in $f\equiv n$ (although this is \textit{not} an admissible function) gives the $n$th bottleneck.

For two Prokhorov-close PDs, we require the number (=counting measure) of unmatched points to be small.
Points with small persistence get matched to the diagonal and thus do not blow up the Prokhorov distance.
Hence it is robust with respect to noise.

\begin{exmp}
Assume $f$ is invertible.
Recall the situation of Example \ref{exmp:singletons}: $X=\{x\}$ and $Y=\{y\}$ both consist of one point each and we assume that $d(x,y)<d(x,x')+d(y,y')$, where the prime denotes the projection to the diagonal.
We saw that the bottleneck profile looks as follows:
\[
    D_{X,Y}(t) = \begin{cases}
        1 & \text{if } 0\le t \le d(x,y),\\
        0 & \text{if } t> d(x,y).
    \end{cases}
\]
It follows that
\[
    \pi_f(X,Y) = \min(f^{-1}(1), d(x,y)).
\]
\end{exmp}

\begin{lem}\label{lem:cadlag}
For $f$ admissible, $D_{X,Y}(\pi_f(X,Y)) \le f(\pi_f(X,Y))$.
\end{lem}
\begin{proof}
Note that $D_{X,Y}$ is right-continuous by construction.

\end{proof}

The triangle inequality follows from Lemma \ref{lem:TriangleD}.
\begin{lem}\label{lem:TriangleIneq}
Fix an admissible function $f\colon \R_{\ge 0} \to \R_{\ge 0}$. For any three persistence diagrams $X,Y,Z$, we have 
\[
\pi_f(X,Z) \le \pi_f(X,Y)+\pi_f(Y,Z).
\]
\end{lem}
\begin{proof}
We make the following estimates:
\begin{align*}
    D_{X,Z}(\pi_f(X,Y) + \pi_f(Y,Z)) &\le D_{X,Y}(\pi_f(X,Y))+D_{Y,Z}(\pi_f(Y,Z))\\
    & \le f(\pi_f(X,Y))+f(\pi_f(Y,Z))\\
    & \le f(\pi_f(X,Y)+\pi_f(Y,Z)).
\end{align*}
Here we used Lemma \ref{lem:TriangleD} for the first inequality, Lemma \ref{lem:cadlag} for the second and superadditivity of $f$ for the final one.
Therefore,
\[
    \inf \{t>0 \colon D_{X,Z}(t)<f(t)\} \le \pi_f(X,Y) + \pi_f(Y,Z);
\]
the left hand side is the definition of $\pi_f(X,Z)$, as desired. 
\end{proof}

As the symmetry is clear, we have shown:
\begin{thm}\label{thm:ProkhorovIsMetric}
Fix an admissible function $f\colon \R_{\ge 0} \to \R_{\ge 0}$. The discrete $f$-Prokhorov metric is an extended pseudometric.
\end{thm}

Just like for the bottleneck distance, we need some finiteness property for the $\pi_f$ to be a genuine metric.
Let $\Bar{\mathcal{B}}$ denote the persistence diagrams which for every $\eps>0$ have only finitely many points of persistence $>\eps$.
Then Lemma \ref{lem:D=0impliesX=Y} implies:
\begin{lem}\label{lem:GenPrDefininite}
Let $f\colon \R_{\ge 0} \to \R_{\ge0}$ be admissible. 
For $X,Y\in \Bar{\mathcal{B}}$, we have $\pi_f(X,Y) = 0$ only if $X=Y$.
\end{lem}
\begin{proof}
If $\pi_f(X,Y) = 0$, then $D_{X,Y}(t)<f(t)$ for all $t>0$.
As the bottleneck profile is monotonically decreasing and $\lim \limits _{t\searrow 0} f(t) = 0$, this implies $D_{X,Y}(t)=0$ for all $t>0$.
By Lemma \ref{lem:D=0impliesX=Y}, this happens only if $X=Y$.%
\end{proof}

Our next task is to investigate how $\pi_f$ depends on the function $f$.
While from a metric point of view, we need to fix $f$, the context of data science suggests a different perspective:
For given training data (a fixed set of persistence diagrams) adjust $f$ to obtain a metric that performs well on it (e.g. in a classification problem, cf. section \ref{subsec:Experiments}).
\begin{lem}\label{lem:MonotoneDecr}
Let $f, g \colon \R_{\ge0} \to \R_{\ge 0}$ such that $f(t)\le g(t)$ for all $t\ge 0$.
Then for any two persistence diagrams $X,Y$, we have $\pi_g(X,Y) \le \pi_f(X,Y)$.
\end{lem}
\begin{proof}
If $t>0$ satisfies $D_{X,Y}(t)<f(t)$, then also $D_{X,Y}(t)<g(t)$.
Therefore, 
 \[
 \inf\{t>0\colon D_{X,Y}(t)<g(t)\} \le \inf\{t>0\colon D_{X,Y}(t)<f(t)\}
 \]
 and by definition $\pi_g(X,Y) \le \pi_f(X,Y)$.
\end{proof}
For fixed persistence diagrams, the Prokhorov metric is continuous with respect to the functions in supremum metric.
\begin{prop}\label{lem:PiContinuous}
Fix two persistence diagrams $X,Y$. Let $f\colon \R_{\ge 0} \to \R_{\ge 0}$ be admissible.
Theen for all $\eps>0$ there is $\delta>0$ such that for each admissible $g\colon \R_{\ge0}\to \R_{\ge 0}$, wwe have
\[
\|f-g\|_\infty < \delta \Rightarrow \lvert\pi_f(X,Y) - \pi_g(X,Y)\rvert < \eps. 
\]
\end{prop}
\begin{proof}
Without loss of generrality, assume that $f(\pi_f(X,Y)) \leq g(\pi_g(X,Y))$ (otherwise exchange $f$ and $g$ below).
This implies $\pi_f(X,Y) \geq \pi_g(X,Y)$ by monotonicity of $D_{X,Y}$.
We choose $\delta < f(\eps)$ and estimate
\begin{align*}
    f(\pi_g(X,Y) + \eps) & \geq f(\pi_g(X,Y)) + f(\eps) & \textnormal{by superadditivity}\\
            &>f(\pi_g(X,Y)) + \delta & \textnormal{by choice of }\delta\\
            &>f(\pi_g(X,Y)) + \|f-g\|_\infty &\textnormal{by choice of }g\\
            &\geq g(\pi_g(X,Y)) & \textnormal{by definition of the sup-norm}\\
            & \geq f(\pi_f(X,Y)) &\textnormal{by assumption.}
\end{align*}
By monotonicity of $f$ we find that
\[
    \pi_f(X,Y) - \pi_g(X,Y) = \vert \pi_f(X,Y) - \pi_g(X,Y)\vert < \eps.
\]
\end{proof}
From a data science perspective, the preceding Lemma allows us to tune the parameter function $f$ on a fixed training set of persistence diagrams.

\subsection{Comparison with Wasserstein}\label{sec:CompWWasserstein}
Fix a persistence diagram $X$ and consider Wasserstein metrics and Prokhorov distances to some other diagram $Y$.
We can perturb $Y$ by adding more ``noise''.
More precisely, we add $k$ points whose distance to the diagonal is less than $\pi_f(X,Y)$ and denote this diagram by $Y_k$.
This does not affect the Prokhorov metric at all, while for all $p \in [1,\infty)$, the value of $W_p(X,Y_k)$ goes to infinity when $k$ does.
This is what we mean when we say that the Prokhorov metric is more robust with respect to noise compared to the Wasserstein metric.
In other (more mathematical) words, the identity map $\id \colon (\Dgm, \pi_f)\to (\Dgm, W_p)$, where $Dgm$ is the set of all persistence diagrams, is nowhere continuous for $p\in[1,\infty)$\footnote{To avoid such problems, one usually restricts to a subset of $Dgm$ of diagrams with ``finite $p$th moment'' \cite{Mileyko2011Probability}  when using $p$-Wasserstein distances.}.
In this section, we further explore the relation between Prokhorov and Wasserstein distances.

Similarly to the proofs in \cite{gibbs} for the measure-theoretic variants, we can bound our metric in terms of the Wasserstein distance.
As we will explain, the metrics $\pi_{t\mapsto t^q}$ are of special interest.
\begin{prop}\label{prop:ProkhorovLeWasserstein}
Let $p\ge 1, q\geq0, c>0$ and $f(t) = c \cdot t^q$. For two persistence diagrams $X,Y$ we have
\[
    \pi_f(X,Y) \le W_p(X,Y)^{\frac{p}{p+q}}\cdot c^{\frac{-1}{p+q}}.
\]
\end{prop}
\begin{proof}
Recall from Lemma \ref{lem:DWasserstein} that 
\[
    D_{X,Y}(t) = \inf\limits_{\eta}\lvert\{x \colon d(x,\eta(x)) >t\}\rvert \leq \frac{1}{t^p}W_p(X,Y)^p.
\]
We now want to find a suitable value of $t$ such that $D_{X,Y}(t)<c \cdot t^q$ to infer that $\pi_f(X,Y)\leq t$.
Plugging in $t = W_p(X,Y)^{\frac{p}{p+q}} \cdot c^{\frac{-1}{p+q}}$, one obtains
\begin{align*}
    \inf\limits_{\eta}\lvert\{x \colon d(x,\eta(x)) >W_p(X,Y)^{\frac{p}{p+q}}\cdot c^{\frac{-1}{p+q}}\}\rvert 
        &\leq \frac{W_p(X,Y)^p}{W_p(X,Y)^{\frac{p^2}{p+q}}\cdot c^{\frac{-p}{p+q}}}
\end{align*}
Now if $q=0$, the right hand side simplifies to $c = f(W_p(X,Y)\cdot c^{-1/p})$.
If $q>0$, we compute
\begin{align*}
    \frac{W_p(X,Y)^p}{W_p(X,Y)^{\frac{p^2}{p+q}}\cdot c^{\frac{-p}{p+q}}}
    & = W_p(X,Y)^{\frac{p^2+pq}{p+q}-\frac{p^2}{p+q}} \cdot c^{\frac{p+q-q}{p+q}}\\
        &= c\cdot\left(W_p(X,Y)^{\frac{p}{p+q}}\cdot c^{\frac{-1}{p+q}}\right)^q.
\end{align*}
Thereforee,
\begin{align*}
    \inf_\eta \lvert\{x \colon d(x,\eta(x)) >W_p(X,Y)^{\frac{p}{p+q}}\cdot c^{\frac{-1}{p+q}}\}\rvert 
    &\leq c\cdot\left(W_p(X,Y)^{\frac{p}{p+q}}\cdot c^{\frac{-1}{p+q}}\right)^q\\
    &= f\left(W_p(X,Y)^{\frac{p}{p+q}}\cdot c^{\frac{-1}{p+q}}\right)
\end{align*}
and we conclude $\pi_f(X,Y) \le W_p(X,Y)^{\frac{p}{p+q}} \cdot c^{\frac{-1}{p+q}}$ as desired.
\end{proof}
\begin{cor}
Let $p\geq 1$, $q\geq 0$ and $c>0$. 
The map $\id \colon (Dgm, W_p) \to (Dgm, \pi_{c\cdot t^q})$ is continuous.
\end{cor}
When comparing with the bottleneck distance, i.e. $p=\infty$ in the above setting, we can say even more:
\begin{prop}\label{prop:ProkhorovLeqBottleneck}
For all admissible $f$ and all persistence diagrams we have $\pi_f(X,Y) \leq W_\infty(X,Y)$.
\end{prop}
\begin{proof}
We recall by Lemma \ref{lem:DBottleneck},
\begin{align*}
    \inf_\eta\{x\colon d(x,\eta(x)) > W_\infty(X,Y)\} = 0 \leq f(W_\infty(X,Y)),
\end{align*}
    and therefore $\pi_f(X,Y)\leq W_\infty(X,Y)$.
\end{proof}
Specializing to $c=1$ and $p\in\{1,\infty\}$ or $q\in\{0,1\}$, we obtain:
\begin{cor}
The following inequalities hold:

    \bgroup
    \def\arraystretch{2}%
    \begin{tabular}{|l||*{3}{c|}}\hline
        \backslashbox{p}{q}
        &\makebox[3em]{$0$}&\makebox[3em]{$1$}&\makebox[3em]{$q$}\\\hline\hline
        $1$ & $d_B \le W_1$ & $\pi \le \sqrt{W_1}$ & $\pi_{t^q} \le W_1^{\frac{1}{1+q}}$\\\hline
        $\infty$ &$d_B\le d_B$& $\pi\le d_B$& $\pi_{t^q} \le d_B$\\\hline
        $p$ & $d_B \le W_p$ & $\pi \le W_p^{\frac{p}{p+1}}$ & $\pi_{t^q} \le W_p^{\frac{p}{p+q}}$\\\hline
    \end{tabular}
    \egroup
\end{cor}
In particular, the Bottleneck Stability Theorem \ref{thm:ChazalStability} implies stability for the new metrics by Proposition \ref{prop:ProkhorovLeqBottleneck}:
\begin{thm}\label{thm:ProkhorovStability}
Let $X,Y$ be finite metric spaces, fix some admissible function $f$ and $k \in \N$. 
Then we have
    \[
    \pi_f(\Dgm(PH_k(X)), \Dgm(PH_k(Y))) \le 2 d_{GH}(X,Y),
    \]
    where $d_{GH}$ is the Gromov-Hausdorff distance.
\end{thm}

We can provide not only lower but also upper bounds for Wasserstein distances in terms of the Prokhorov distance.
\begin{prop}\label{prop:WassersteinLeProkhorov}
$W_q(X,Y)^q \leq \pi_{t^q}(X,Y)^q (\max(d(x,\eta(x)))^q+\vert \eta\vert ) $, where $\eta\colon X \to Y$ is any matching realizing $\pi_{t^q}(X,Y)$.
\end{prop}
\begin{proof}
For an arbitrary bijection $\eta \colon X \to Y$, consider any $t>0$ such that $\lvert\{d(x,\eta(x))>t\}\rvert \le t^q$.
We estimate:
\begin{align*}
    W_q(X,Y)^q &\leq \sum_{x} d(x,\eta(x))^q\\
    &=\sum_{d(x,\eta(x))>t}d(x,\eta(x))^q+\sum_{d(x,\eta(x))\le t}d(x,\eta(x))^q\\
    &\le \lvert\{d(x,\eta(x))>t\}\rvert \max(d(x,\eta(x)))^q+t^q \lvert\{d(x,\eta(x))\le t\}\rvert\\
    &= \lvert\{d(x,\eta(x))>t\}\rvert \max(d(x,\eta(x)))^q+t^q (\lvert \eta \rvert-\lvert\{d(x,\eta(x))>t\}\rvert)\\
    &= \lvert\{d(x,\eta(x))>t\}\rvert (\max(d(x,\eta(x)))^q-t^q) + t^q \lvert \eta \rvert\\
    &\le t^q \max(d(x,\eta(x)))^q -t^{2q} +t^q \lvert\eta\rvert
\end{align*}
Taking the infimum over all matchings and all such $t$ we obtain the desired inequality
\[
    W_q(X,Y)^q \le \pi_{t^q}(X,Y)^q(\max(d(x,\eta(x)))^q+\lvert \eta \rvert).
\]
\end{proof}
Combining the two inequalities from Propositions \ref{prop:ProkhorovLeWasserstein} and \ref{prop:WassersteinLeProkhorov}, we obtain a comparison for different Wasserstein metrics.
\begin{cor}\label{cor:WassersteinLeWasserstein}
$W_q(X,Y)^q \le W_p(X,Y)^{\frac{pq}{p+q}}(\max(d(x,\eta(x)))^q+\vert \eta\vert )$.
\end{cor}

\begin{rem}
Another inequality relating Wasserstein distances for different $p$ and $q$ ori{-}ginates from the H\"o{}lder inequality, given in \cite[Lemma 3.5]{atienzaStabilityPersistentEntropy2020a}:
For finite persistence diagrams $X$, $Y$ and real numbers $1\leq q < p  < \infty$, we have
\[
    W_q(X,Y) \leq \vert\eta\vert^{\frac{1}{q}-\frac{1}{p}} W_p(X,Y),
\]
where $\eta$ is the matching realizing $W_p(X,Y)$.
Our inequality above yields a lower exponent for $W_p(X,Y)$ at the cost of multiplying with the largest distance in the matching.
In particular, for $q=1, p=2$, our formula reads
\[W_1(X,Y) \le W_2(X,Y)^{\frac{2}{3}}(\max(d(x,\eta(x)))+\vert \eta\vert),\]
with $\eta$ realizing $\pi_{t^q}(X,Y)$, whereas the one of \cite{atienzaStabilityPersistentEntropy2020a} reads (with $\eta$ realizing $W_2(X,Y)$)
\[
    W_1(X,Y) \le W_2(X,Y) \vert \eta \vert^{\frac{1}{2}}.
\]
Depending on the size of $W_p(X,Y)$ relative to the size of $X$ and $Y$, our inequality can provide sharper bounds than the one of \cite{atienzaStabilityPersistentEntropy2020a}.
To investigate the size of $\max(d(x,\eta(x)))$ remains an interesting question for future work.
One possible application of such inequalities is that they allow to infer stability results for vectorizations with respect to $W_p$ for $p>1$ from the stability with respect to $W_1$.
Another use of Propositions \ref{prop:ProkhorovLeWasserstein} and \ref{prop:WassersteinLeProkhorov} is that the bounds they provide for Wasserstein distances are  easily computed, as we will see in Section \ref{sec:Algos} below.
\end{rem}

\subsection{Metric and Topological Properties}\label{subsec:MetricProperties}
Using the comparison with Wasserstein (Section \ref{sec:CompWWasserstein}) and the results from \cite{Mileyko2011Probability}, we address questions of convergence and separability. We run into similar issues as \cite[Theorems 4.20, 4.24, 4.25]{bubenikTopologicalSpacesPersistence2018}  and \cite[section 3]{blumbergRobustStatisticsHypothesis2014}.
In this section, we explicitly allow diagrams with a countably infite number of off-diagonal points under certain finiteness assumptions specified below.

\begin{thm}
Let $p \geq 1$.
The space of persistence diagrams with finite $p$th moment endowed with the $c \cdot t^q$-Prokhorov metric is separable.
\end{thm}
\begin{proof}
Let $\eps>0$, $X$ a persistence diagram and $p\ge 1$.
Let $S$ be a countable dense subset for the $p$-Wasserstein metric; this exists by \cite[Theorem 12]{Mileyko2011Probability}.
In fact they show that we can take $S$ to be the set of finite diagrams whose points have rational coordinates.
Let $X_S\in S$ be a persistence diagram such that $W_p(X, X_S) < \eps^{\frac{p+q}{p}}\cdot c^{\frac{1}{p}}$.
Then by Proposition \ref{prop:ProkhorovLeWasserstein}, we have
\[
\pi_{c\cdot t^q}(X,X_S)\le W_p(X, X_S)^{\frac{p}{p+q}} \cdot c^{\frac{-1}{p+q}} < \eps^{\frac{p}{p+q}\frac{p+q}{p}} \cdot c^{\frac{-1}{p+q}} \cdot  c^{\frac{1}{p+q}} = \eps.
\]
\end{proof}

Note that the assumptions in the previous Theorem are weaker than the ones usually considered for the bottleneck distance, compare \cite[Theorem 4.18]{bubenikTopologicalSpacesPersistence2018}.

Recall that $\Bar{\mathcal{B}}$ denotes the persistence diagrams which for all $\eps>0$ have finitely many points of persistence $>\eps$.
The next Theorem is a consequence of \cite[Theorem 3.5]{blumbergRobustStatisticsHypothesis2014}, which asserts that the bottleneck distance makes $\Bar{\mathcal{B}}$ into a Polish space.

\begin{thm}
The space $\Bar{\mathcal{B}}$ endowed with the Prokhorov metric $\pi_f$ is Polish for all admissible $f$.
\end{thm}
\begin{proof}
Let $(X_n) \subset \Bar{\mathcal{B}}$ be a Cauchy sequence with respect to the Prokhorov metric $\pi_f$.
Let $\eps>0$ such that $f(\eps)\leq 1$.
Then the inequality $\pi_f(X_m,X_n)<\eps$ implies by definition of $\pi_f$ that
\[
    D_{X_m, X_n}(\eps)<f(\eps)\leq  1.
\]
As the bottleneck profile takes values in the integers, we conclude that $D_{X_m, X_n}(\eps) = 0$ and hence, by Lemma \ref{lem:DBottleneck}, we have $\eps \geq W_\infty(X_m,X_n)$.
In particular, $X_n$ is a Cauchy sequence with respect to the bottleneck distance.
By completeness of $\Bar{\mathcal{B}}$ with the bottleneck distance, there is a limit diagram $X\in \Bar{\mathcal{B}}$ to which the sequence converges.
Finally by Lemma \ref{prop:ProkhorovLeqBottleneck}, convergence in bottleneck implies convergence in Prokhorov.

Now for separability, consider a subset $A\subset\Bar{\mathcal{B}}$ which is dense with respect to the bottleneck distance.
Let $X\in \Bar{\mathcal{B}}$ and $\eps>0$. 
Then by assumption, there is $Y\in A$ with $W_\infty(X,Y)<\eps$.
Then, since by Proposition \ref{prop:ProkhorovLeqBottleneck} $ \pi_f(X,Y) \leq W_\infty(X,Y)$, we also have $\pi_f(X,Y) < \eps $.
Therefore, $A$ is dense in $\Bar{\mathcal{B}}$ with respect to $\pi_f$ as well.
\end{proof}

\subsection{Algorithms}\label{sec:Algos}
In this section, all persistence diagrams are finite.
Now we will provide an algorithm to compute $\pi_f(X,Y)$ for continuous monotonically increasing functions $f$.
In this case, there is always a single value $t_0\in [0,\infty)$ such that $D_{X,Y}(t)<f(t)$ for $t > t_0$ and $D_{X,Y}(t)>f(t)$ for $t<t_0$.
We can find its location by bisection.
Recall that we set $n = \vert X\vert + \vert Y \vert$.
\begin{prop}\label{prop:ProkhorovAlgo}
Let $f\colon [0,\infty)\to [0,\infty)$ be monotonically increasing. Assume that the values and preimages of $f$ can be computed in $O(1)$.
Then $\pi_f(X,Y)$ can be computed in $O(n^2 \log(n))$.
\end{prop}
\begin{proof}
First, observe that the Prokohorv distance takes its value among the pairwise distances of points in the persistence diagrams (if $f$ crosses the bottleneck profile at one of its vertical gaps) or among preimages of integers under $f$ (if $f$ crosses the bottleneck profiles at one of its constant pieces), in formulas 
\[
\pi_f(X,Y) \in \{d(x,y)\colon x\in X, y\in Y\} \cup f^{-1}(\N_{\le \vert X\vert +\vert Y\vert }) =\colon T_1.
\]
To perform a binary search, we sort the elements in $T_1$ as a preprocess, which has runtime complexity $O(n^2 \log(n))$.
In each iteration of the binary search we pick the median $t\in T_i$.
Next we compute the value of the bottleneck profile $D_{X,Y}(t)$ using Proposition \ref{prop:DRuntime}, taking $O(n^2)$.
Then we compute $f(t)$, which by assumption takes $O(1)$.
Now if $D_{X,Y}(t)>f(t)$ set $T_{i+1}$ to be the right half, if $D_{X,Y}(t)\le f(t)$ set $T_{i+1}$ to be the left half of $T_i$.
Hence we obtain a runtime of $O(n^2 \log{n})$ for the binary search as well.
\begin{algorithm}[H]
 \caption{The binary search to compute $\pi_f(X,Y)$}
\begin{algorithmic}[1]
\Require{Persistence diagrayms $X,Y$; function $f$}
\Ensure{$\pi_f(X,Y)$ }
 \State $T = \{d(x,y)\colon x\in X, y\in Y\} \cup f^{-1}(\N_{\le \vert X\vert+\vert Y\vert })$
 \State sort $T$
 \State $L = 0; R = \text{length}(T)$
 \While{$L<R$}
    \State $m = \lfloor \frac{R+L}{2}\rfloor$
    \State $t = T[m]$
    \If{$D_{X,Y}(t)>f(t)$}
        \State$L=m+1$
    \Else
        \State $R=m$
    \EndIf
 \EndWhile\\
 \Return{$T[L]$}
 \end{algorithmic}
\end{algorithm}
\end{proof}
In particular, if one uses a more efficient geometric data structure to improve the runtime of the matching algoritthm, the sorting preprocessing dominates the runtime.
Compare \cite{Efrat2001Geometry}, Theorem 3.2 and the preceding discussion therein for more details and possible improvements of the runtime complexity.
Please refer to Section \ref{sec:implementation} for details about our implementation and its availability.

There is an easy modification to the above algorithm to approximate $\pi_f$ up to an additive error of $\eps$.
Instead of performing the binary search on the indicated discrete set (which needs to be sorted or otherwise pre-processed in a costly way, as noted), one can run it on an interval $[0,M]$.
Here, $M$ is some upper bound, for example the sum of the longest lifespans of points in $X$ and $Y$ respectively (which is computed in $O(n)$).
We bisect the interval until we arrive at one of length less than $2\eps$.
Its midpoint is guaranteed to be less than $\eps$ away from the true value of $\pi_f(X,Y)$.

\section{Experiments}\label{subsec:Experiments}
A simple application of the bottleneck profile, based on simple synthetic persistence diagrams, was already presented in Example \ref{exmp:bottleneckprofile}. 

\subsection{Highlighting Geometric Intuition}\label{sec:three_circles}
This experiment is a toy example, showing how the Prokhorov distance can capture our geometric intuition more accurately than bottleneck or Wasserstein.
Consider three different shapes in $\mathbb{R}^2$: a) a big circle ($r=6$), b) a big ($r=6$) and a medium circle ($r=4$), c) a big ($r=6$), a medium ($r=4$) and small circle ($r=2$).
We take five samples with noise from each shape according to Table \ref{tab:3Circles}.
\begin{table}[!htbp]
    \centering
        \begin{adjustwidth}{-1.5in}{-1.5in}
    \centering
    \begin{tabular}{|c|c|c|c|c|c|}\hline
         shape  & number of circles & radii & samples & noise                             & colour in the figures \\ \hline
         a      & 1 & 6       & 120    & uniform from $[-0.2,0.2]^2$ & blue\\
         b      & 2 & 6, 4    & 300    & uniform from $[-0.23,0.23]^2$ & red\\
         c      & 3 & 6, 4, 2 & 120    & uniform from $[-0.2,0.2]^2$ & green\\ \hline
    \end{tabular}
    \caption{The three shapes: one two and three circles.}
    \label{tab:3Circles}
    \end{adjustwidth}
\end{table}

\begin{figure}[!htbp]
    \centering
    \includegraphics[width=0.8\textwidth]{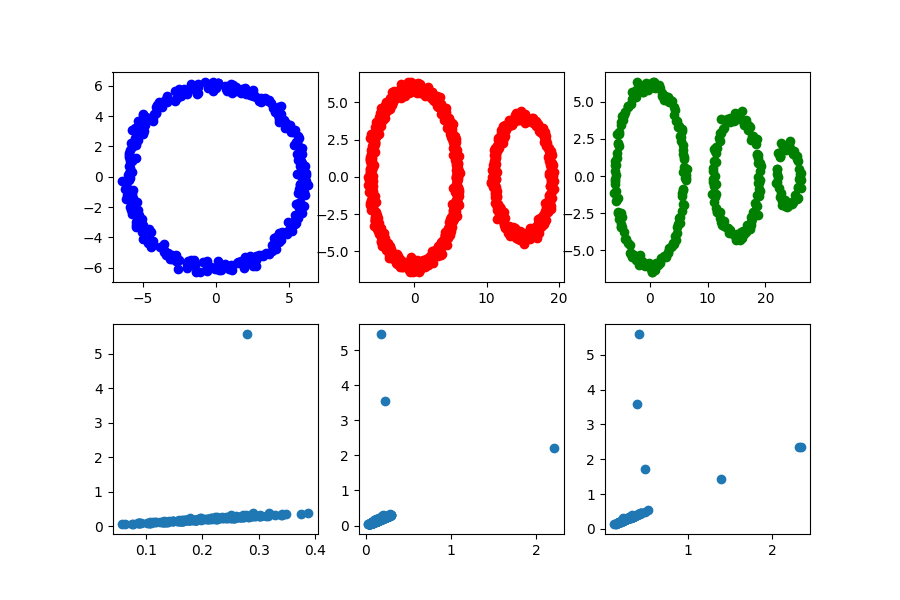}
    \caption{\label{fig:ThreeCircles}One two and three noisy circles and their PDs for the first persistent homology.}
    
\end{figure}

For each point cloud we compute the first persistent homology modules of it alpha complex filtration and represent them as PDs (see Figure~\ref{fig:ThreeCircles}).
We can look at the averaged $D$-function for each pair of shapes (Figure \ref{fig:DThreeCircles}).
After careful inspection of this figure and some trial and error, we come up with the choice of $f(t) = t^3\cdot 20^t$ to separate three bottleneck profiles in a most efficient way:
Between around $0.55$ and $0.65$, the averaged bottleneck profiles involving shape c) with the small circle decrease, while the one comparing a) and b) stays constant.
Intersecting with a function in this interval will provide a good choice for the Prokhorov distance:
It puts the two and three circles closest to each other and one and three circles the farthest apart.
In data science tasks, we will of course need an automated way to find a good parameter function $f$, we will discuss this in more detail below.

\begin{figure}[!htbp]
    \centering
    \includegraphics[width=0.8\textwidth]{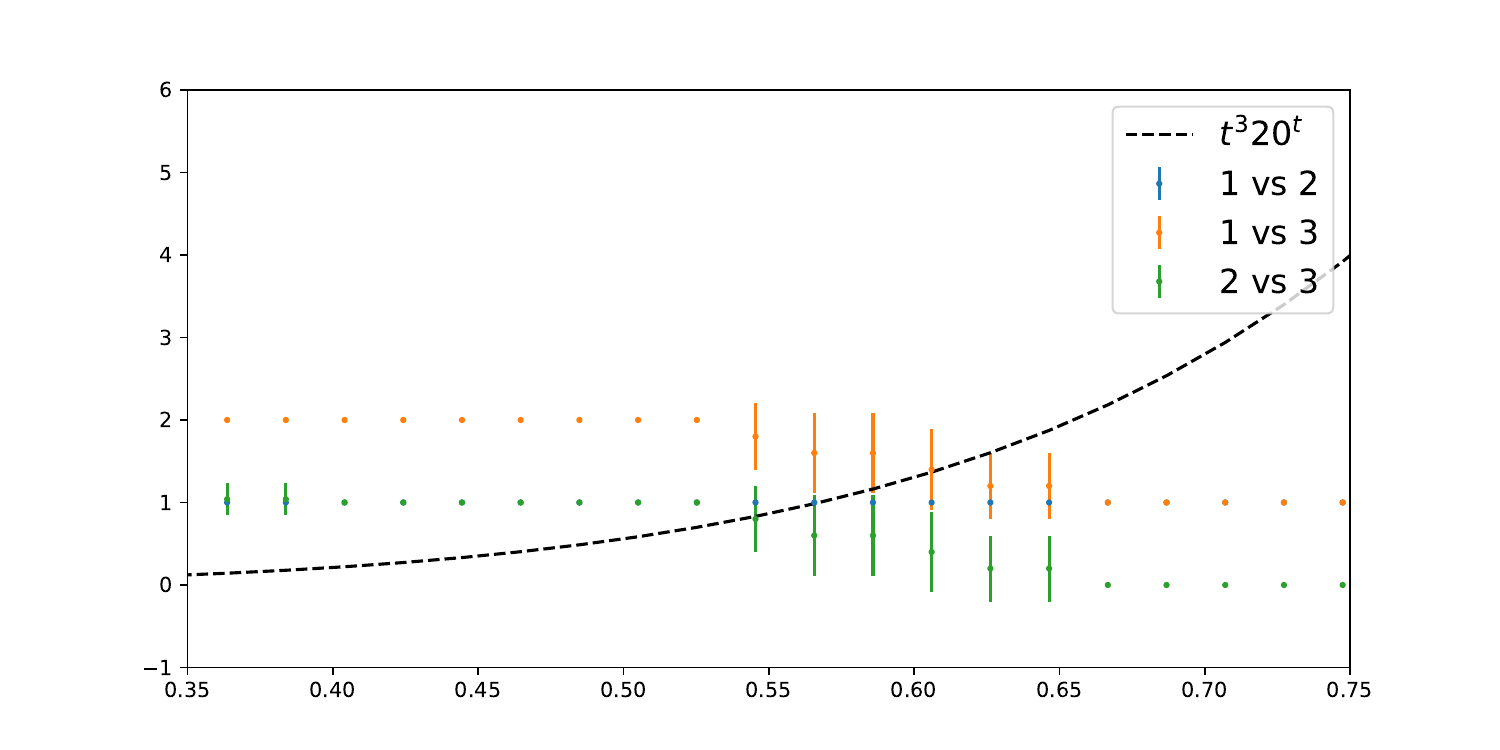}
    \caption{\label{fig:DThreeCircles}The averaged bottleneck profile for the three circles.}

\end{figure}

Now we want to compare the Bottleneck, Prokhorov and Wasserstein distances.

\begin{figure}[!htbp]
    \centering
    \begin{subfigure}[t]{0.3\textwidth}
	    \centering
	    \includegraphics[width=\textwidth]{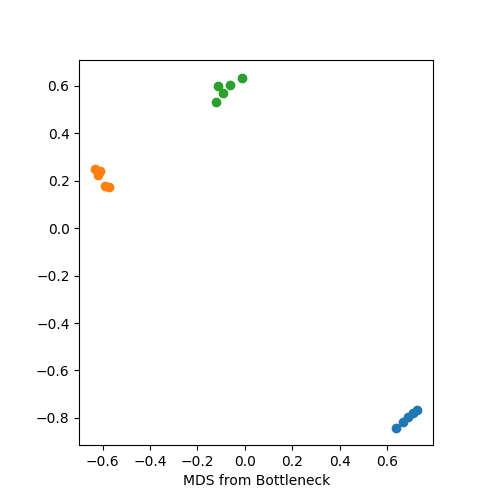}
	\end{subfigure} 
	\hfill
	\begin{subfigure}[t]{0.3\textwidth}
	    \centering
	    \includegraphics[width=\textwidth]{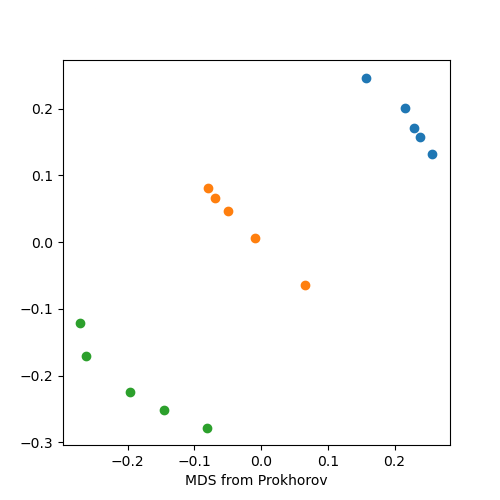}
	\end{subfigure}
	\hfill
	\begin{subfigure}[t]{0.3\textwidth}
	    \centering
	    \includegraphics[width=\textwidth]{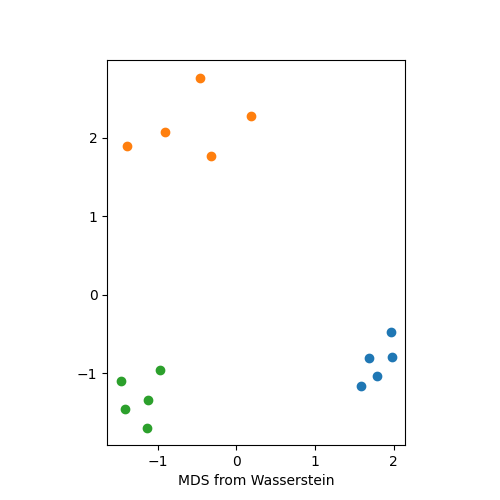}
	\end{subfigure}
	\caption{\label{fig:MDS}MDS plots of the dataset in Section~\ref{sec:three_circles}.}
	\vspace*{\floatsep}
    \includegraphics[width=\textwidth]{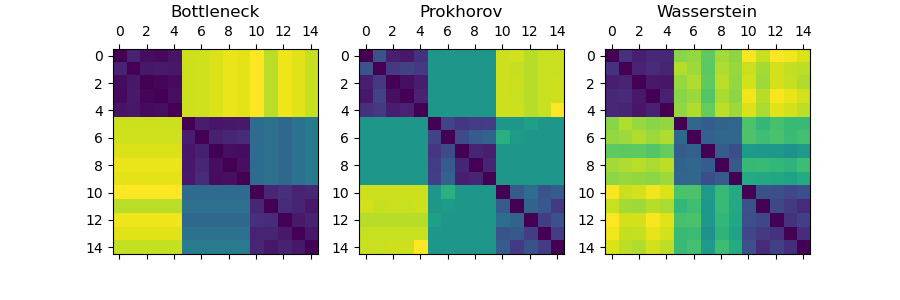}
    \label{fig:ThreeCirclesDMatrix}
    \caption{Distance matrices of the dataset in Section~\ref{sec:three_circles}.}
    
\end{figure}
The bottleneck distance between shapes a) and both b) and c) is roughly the same.
This distance does not take the presence of the additional small circle in shape c).
By blowing up the sample size and the noise in shape b), the Wasserstein distance from a) and c) to it are artificially blown up (Figures \ref{fig:MDS} and \ref{fig:ThreeCirclesDMatrix}).
The Prokhorov distance is built to avoid these pitfalls and nicely captures the geometry of the setting.
The MDS plot for Prokhorov agrees with our intuition and places b) between a) and c) (Figures \ref{fig:MDS}).

\subsection{Classification Experiments}
We now turn to more sophisticated data sets to illustrate the usage and advantages of the Prokhorov distance.
In particular, we consider persistence diagrams that actually arise in applications of TDA.
We use the library \cite{scikit-learn} for standard machine learning algorithms (in particular $K$-Neighbors). 
For the Bottleneck and Wasserstein metrics we use the Gudhi library \cite{gudhi:BottleneckDistance} and \cite{gudhi:PersistenceRepresentations}.
To score the different metrics, we use K-neighbors classification accuracy as well as classification accuracy based on K-Medoids clustering with the ``build'' initialization \cite{schubert_erichschuberttu-dortmundde_kmedoids_nodate}, \cite{schubert_fast_2020}.
In the latter case, points are assigned to the class of the medoid of their cluster.
We split the data sets into training and testing with $50\%$ of the points each.
All computations were carried out on a laptop with an Intel i5-8265U CPU with $1.60$ GHz and $8$ GB memory.
The code to reproduce the experiments is available online\footnote{ \url{https://github.com/nihell/ProkhorovExamples/blob/master/Experiments.ipynb}}.

\subsubsection*{Parameter tuning -- choosing $f$}
One needs to specify an admissible function $f$ as a parameter for the Prokhorov distance $\pi_f$.
The set off all such functions is vast, therefore it is sensible to restrict to a smaller subset.
In the experiments below, we choose $f$ from linear functions with integer slope $\in [10,100]$.
We do this by performing a grid search over the parameters and evaluating them by five-fold cross-validation.
By selecting this subset of parameters, we reduce the risk of overfitting and are able to run the parameter selection in reasonable time.
We leave it as a problem for further investigation to find better means to run the parameter selection, but note that the fact that the bottleneck profile is piecewise constant obstructs the use of gradient descent.

\subsubsection*{Prokhorov Distance for Cubical Complexes with Outlier Pixels}
We generate\footnote{Code available at \url{https://github.com/nihell/ProkhorovExamples/blob/master/GenerateCubicalNoise.ipynb}} $100 \times 100$ pixel greyscale images according to the following procedure, cf. Figure \ref{fig:CubicalComplex}.
Initializinng every pixel with $0$, we choose $n$ points at random, at which we add a Gau{\ss}ian with $\sigma = 3$.
We normalize the values to $[0,2]$ and then shift them up by $64$.
The goal is to distinguish images with $n=15$ from images with $n=20$.
The obstacle is that we superimpose a particular kind noise, similar to salt-and-pepper noise.
We choose $k$ pixels randomly at which we set the value to a random integer from $[1,128]$; the eight surrounding pixels are set to zero.
\begin{figure}
    \centering
    \begin{subfigure}[b]{0.3\textwidth}
    \centering
        \includegraphics[width=\textwidth]{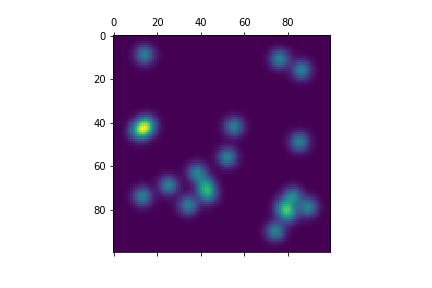}
    \end{subfigure}
    \begin{subfigure}[b]{0.3\textwidth}
        \centering
        \includegraphics[width=\textwidth]{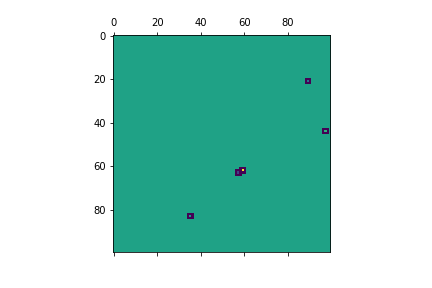}
    \end{subfigure}
    \begin{subfigure}[b]{0.3\textwidth}
        \centering
        \includegraphics[width=\textwidth]{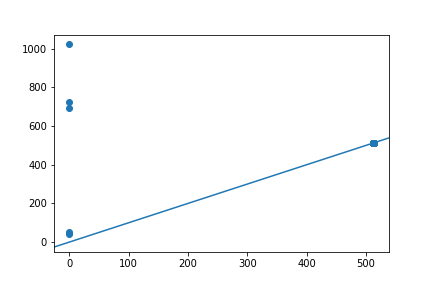}
    \end{subfigure}
    \caption{The underlying Gau\ss{}ians, the superimposed noise and the resulting persistence diagram\label{fig:CubicalComplex}}
\end{figure}
For each of the four combinations $n\in\{15,20\}$ and $k\in\{3,5\}$ we sample $50$ greyscale images. 
We then create a cubical complex from each using the pixels as top-dimensional cells (lower-star filtration) and compute persistent homology in dimensions $0$ and $1$.
We proceed as indicated at the beginning of this section to assess the accuracy of the different metrics.
The results are summarized in Table \ref{tab:CubicalNoiseResults}.
Both in dimmension 0 and 1, the K-Neighbors classifier is inconclusive in the setting of Bottleneck and Wasserstein.
With a suitable Prokhorov metric, we are able to achieve an accuracy of more than $80\%$.
In the K-Medoids approach, the story is similar but less pronounced:
Bottleneck and Wasserstein are inconclusive, but Prokhorov achieves around $60\%$ accuracy.

\begin{table}[!htbp]
    \centering
    \begin{adjustwidth}{-1.5in}{-1.5in}
    \centering
   \begin{tabular}{|c|c|c||c|c||c||c||c|}
    \hline
                              & dim                  &  $f(t)$               & Prokhorov & Bottleneck & 1-Wasserstein & 2-Wasserstein\\ \hline
    K neighbors training score& \multirow{4}{*}{$0$} & \multirow{3}{*}{$49t$} & $0.8425 $& $ 0.515 $& $ 0.58 $& $ 0.525 $\\
    K neighbors test score    &                      &                       & $0.8575 $& $ 0.485 $& $ 0.535 $& $ 0.4925 $\\
    computation time $[s]$    &                      &                       &  $ 26.69$ &$ 36.61 $& $ 44.56 $& $ 125.2 $\\ 
    parameter tuning time $[s]$&                      &  1059               &          &           &               &       \\    \hline
    K medoids training score&   \multirow{4}{*}{$0$} & \multirow{3}{*}{$18t$} & $0.5975 $& $ 0.5325 $& $ 0.515 $& $ 0.51 $\\
    K medoids test score    &                      &                       & $0.62 $& $ 0.51 $& $ 0.485 $& $ 0.5125 $\\
    computation time $[s]$    &                      &                       &  $ 70.14$ &$ 108.6 $& $ 127.6 $& $ 379.4 $\\ 
    parameter tuning time $[s]$&                      &  1082               &          &           &               &       \\    \hline
    
     K neighbors training score& \multirow{4}{*}{$1$} & \multirow{3}{*}{$92t$} & $0.8625 $& $ 0.5025 $& $ 0.545 $& $ 0.4825 $\\
    K neighbors test score    &                      &                       & $0.825 $& $ 0.5375 $& $ 0.575 $& $  0.495 $\\
    computation time $[s]$    &                      &                       &  $ 26.57$ &$ 36.68 $& $ 44.43 $& $ 125.9 $\\ 
    parameter tuning time $[s]$&                      &  1025               &          &           &               &       \\    \hline
    K medoids training score&   \multirow{4}{*}{$1$} & \multirow{3}{*}{$16t$} & $0.62 $& $ 0.4925 $& $ 0.49 $& $ 0.485 $\\
    K medoids test score    &                      &                       & $0.5975 $& $ 0.5125 $& $ 0.4975 $& $ 0.515 $\\
    computation time $[s]$    &                      &                       &  $ 77.08$ &$ 113.9 $& $ 132.6 $& $ 401.9 $\\ 
    parameter tuning time $[s]$&                      &  1098               &          &           &               &       \\    \hline

\end{tabular}
\end{adjustwidth}
\caption{Classification scores for the synthetic dataset.}
    \label{tab:CubicalNoiseResults}
\end{table}

\subsubsection*{3D Segmemtation}
We adapt an example from \cite{carriereSKlearnTDA} and \cite{gudhi:PersistenceRepresentations}, which is based on the dataset \cite{Chen:2009:ABF}.
The task is to classify 3D-meshes based on the persistence diagrams of certain functions defined on them.
The shapes are for example airplanes, hands, chairs ...
The results of classification are presented in the Tables \ref{tab:3dSegResults}.
All the considered metrics yield a similar accuracy.
Prokhorov is the fastest, however at the cost of first having to find the suitable parameter, which took moore than ten hours in this case.

\begin{table}[!htbp]
    \centering
    \begin{adjustwidth}{-1.5in}{-1.5in}
    \centering
   \begin{tabular}{|c|c||c|c||c||c||c|}
    \hline
                              &   $f(t)$               & Prokhorov & Bottleneck & 1-Wasserstein & 2-Wasserstein\\ \hline
    K neighbors training score&  \multirow{3}{*}{$10t$} & $0.9101 $& $ 0.9098 $& $ 0.9042 $& $ 0.9059 $\\
    K neighbors test score    &                        & $0.9270 $& $ 0.9245 $& $ 0.9312 $& $ 0.9298 $\\
    computation time $[s]$    &                        &  $ 795.2$ &$ 1252 $& $ 838.5 $& $ 1740 $\\ 
    parameter tuning time $[s]$&   40440               &          &           &               &       \\    \hline
    K medoids training score&    \multirow{3}{*}{$13t$} & $0.4792 $& $ 0.4905 $& $ 0.4021 $& $ 0.4592 $\\
    K medoids test score    &                         & $0.4985 $& $ 0.4891 $& $ 0.4126 $& $ 0.5125 $\\
    computation time $[s]$    &                       &  $ 1946$ &$ 3467 $& $ 2057 $& $ 5009 $\\ 
    parameter tuning time $[s]$&   41715               &          &           &               &       \\    \hline

\end{tabular}
\end{adjustwidth}
\caption{Classification scores for the 3d segmentation dataset.}
    \label{tab:3dSegResults}
\end{table}

\subsubsection*{Synthetic Dataset}
Finally, we consider the dataset introduced by \cite[Section 6.1]{Adams2017Persistence}.
It contains six shape classes: A sphere, a torus, clusters, clusters within clusters, a circle and the unit cube.
From each class take 25 samples of 500 points.
Then add two levels of Gaussian noise ($\eta = 0.05, 0.1$) and the zeroth and first persistent homology of the Vietoris-Rips filtration are computed.
We compute the distance matrices and evaluate them based on the $K$-neighbors and $K$-medoids classifiers.
The results are displayed in Table \ref{tab:SynthDatasetResults}.
We find that Prokhorov performs better Bottleneck and only slightly worse than Wasserstein.
Prokhorov takes at most similarly long as 1-Wasserstein; Bottleneck is faster and 2-Wasserstein is slower.
\begin{table}[!htbp]
    \centering
    \begin{adjustwidth}{-1.5in}{-1.5in}
    \centering
    \begin{tabular}{|c|c|c||c|c||c||c||c|}
    \hline
                              & dim                  & noise                  &              $f(t)$& Prokhorov & Bottleneck & 1-Wasserstein & 2-Wasserstein\\ \hline
    K neighbors training score& \multirow{4}{*}{$0$} & \multirow{4}{*}{$0.05$}& \multirow{3}{*}{$42t$}& $0.9067 $& $ 0.8133 $& $ 1.0 $& $ 0.9867 $\\
    K neighbors test score    &                      &                        &                       & $0.84 $& $ 0.7867 $& $ 0.96 $& $ 0.9467 $\\
    computation time $[s]$    &                      &                        &                       &  $ 144.2$ &$ 45.39 $& $ 252.8 $& $ 1063 $\\ 
    parameter tuning time $[s]$&                      &                       &  4218                 &          &           &               &       \\    \hline
    K medoids training score&   \multirow{4}{*}{$0$} & \multirow{4}{*}{$0.05$}&\multirow{3}{*}{$93t$} & $0.8 $& $ 0.68 $& $ 0.9733 $& $ 0.8933 $\\
    K medoids test score    &                      &                          &                       & $0.9067 $& $ 0.6 $& $ 0.88 $& $ 0.8933 $\\
    computation time $[s]$    &                      &                        &                       & $ 465.3$ &$ 156.1 $& $  801.3 $& $ 3207 $\\ 
    parameter tuning time $[s]$&                      &                       &4507                   &          &           &               &       \\    \hline
    
    K neighbors training score& \multirow{4}{*}{$0$} & \multirow{4}{*}{$0.1$}& \multirow{3}{*}{$87t$} & $0.9733 $& $ 0.7867 $& $ 0.9867 $& $ 0.9867 $\\
    K neighbors test score    &                      &                        &                       & $1.0 $& $ 0.7467 $& $ 1.0 $& $ 1.0 $\\
    computation time $[s]$    &                      &                        &                       &  $ 145.8$ &$ 44.22 $& $ 267.0 $& $ 1081 $\\ 
    parameter tuning time $[s]$&                      &                       &  4267                 &          &           &               &       \\    \hline
    K medoids training score&   \multirow{4}{*}{$0$} & \multirow{4}{*}{$0.1$}&\multirow{3}{*}{$95t$}  & $0.8 $& $  0.6 $& $ 0.9867 $& $  0.96 $\\
    K medoids test score    &                      &                          &                       & $0.9067 $& $ 0.56 $& $ 0.96 $& $ 0.9733 $\\
    computation time $[s]$    &                      &                        &                       & $ 465.3$ &$ 161.0 $& $  791.4 $& $ 3195 $\\ 
    parameter tuning time $[s]$&                      &                       &4850                   &          &           &               &       \\    \hline
    
    K neighbors training score& \multirow{4}{*}{$1$} & \multirow{4}{*}{$0.05$}& \multirow{3}{*}{$51t$}& $0.9733 $& $ 0.92 $& $ 1.0 $& $ 1.0 $\\
    K neighbors test score    &                      &                        &                       & $0.96 $& $ 0.9333 $& $ 1.0 $& $ 1.0 $\\
    computation time $[s]$    &                      &                        &                       &  $ 24.97 $ &$ 22.82 $& $ 23.77 $& $ 118.5 $\\ 
    parameter tuning time $[s]$&                      &                       &  736.2                 &          &           &               &       \\    \hline
    K medoids training score&   \multirow{4}{*}{$1$} & \multirow{4}{*}{$0.05$}&\multirow{3}{*}{$98t$} & $0.8 $& $0.7867 $& $ 1.0 $& $ 1.0 $\\
    K medoids test score    &                      &                          &                       & $0.8667 $& $ 0.8267 $& $ 1.0 $& $ 1.0 $\\
    computation time $[s]$    &                      &                        &                       & $ 77.63$ &$ 76.08 $& $  72.20 $& $ 366.7 $\\ 
    parameter tuning time $[s]$&                      &                       &779.6                   &          &           &               &       \\    \hline
    
     K neighbors training score& \multirow{4}{*}{$1$} & \multirow{4}{*}{$0.1$}&\multirow{3}{*}{$61t$} & $0.9333 $& $ 0.92 $& $ 0.92 $& $ 0.9333 $\\
    K neighbors test score    &                      &                        &                       & $0.9467 $& $ 0.93333 $& $ 0.9867 $& $  0.9867 $\\
    computation time $[s]$    &                      &                        &                       &  $ 28.17$ &$ 22.28 $& $ 26.98 $& $ 138.4 $\\ 
    parameter tuning time $[s]$&                      &                       & 809.1               &          &           &               &       \\    \hline
    K medoids training score&   \multirow{4}{*}{$1$} & \multirow{4}{*}{$0.1$} &\multirow{3}{*}{$50t$} & $0.88 $& $  0.6933 $& $ 0.8133 $& $ 0.8133 $\\
    K medoids test score    &                      &                          &                       & $0.8133 $& $ 0.7067 $& $ 0.8533 $& $ 0.8533 $\\
    computation time $[s]$    &                      &                        &                       &  $ 88.91$ &$ 75.50 $& $ 80.01 $& $ 413.8 $\\ 
    parameter tuning time $[s]$&                      &                       & 832.2               &          &           &               &       \\    \hline

\end{tabular}
\end{adjustwidth}
\caption{Classification scores for the synthetic dataset from \cite{Adams2017Persistence}.}
    \label{tab:SynthDatasetResults}
\end{table}

\subsection{Discussion}
First and foremost, we found that Prokhorov is able to produce good results in situations where the classical tools of Bottleneck and Wasserstein fail.

In order to explain the differences in the computation time, we note the size of the persistence diagrams in the various settings:

\begin{table}[!htbp]
    \centering
    \begin{adjustwidth}{-1.5in}{-1.5in}
    \centering
   \begin{tabular}{|c|c|c|c|c|c|}
    \hline
     & 3D-Segmentation & Synthetic data  & Synthetic data  & Synthetic data & Synthetic data\\
     & & $H_0$, $\eta = 0.05$ & $H_0$, $\eta = 0.1$ &  $H_1$, $\eta = 0.05$ & $H_1$, $\eta = 0.1$ \\\hline
     Mean size          & $ 11.84 $ &$ 500 $& $ 500 $& $ 177.7 $& $ 189.9 $\\     \hline
     standard deviation & $ 4.893 $ & $ 0 $& $ 0 $& $ 40.53 $& $ 38.84 $\\     \hline
\end{tabular}
\end{adjustwidth}
\caption{Cardinalities of the persistence diagrams for the considered experiments.}
    \label{tab:PDsizes}
\end{table}

By inspecting Table~\ref{tab:PDsizes} wee see that the 3D segmentation dataset contains way smaller diagrams, on which the Prokhorov metric seems to perform well, both in terms of runtime and score.
On the bigger diagrams from the synthetic dataset, the Wasserstein metrics yield the highest scores.
Prokhorov outperforms Bottleneck in the scores at the cost of higher runtimes.
The difference in the computation time is caused by the evaluation of $f(t)$, which is the only difference between the Bottleneck and Prokhorov implementations.

Bottleneck -- and to some extend also Prokhorov -- work less well on zero-dimensional PDs.
There, every class is born at time zero, hence the PD is intrinsically one-dimensional and points are matched in linear order. 
The bottleneck distance is less meaningful in this setting.
Moreover, the Prokhorov (and even more the Bottleneck) distance do not take points matched over a small distance into account.
This is a consequence of being designed to be robust against noise.
However, this data can actually contain meaningful information, which is picked up by the Wasserstein distances.
This is a possible explanation for the fact that Wasserstein yields better scores in the synthetic dataset.

Hence, the Prokhorov metric works best on rather small diagrams and runs fastest with simple (e. g. linear) parameter functions $f$.
Even then, one needs to take the additional time for tuning the parameter $f$ into account.

\section{Discussion and Outlook}
Summarizing the results from the previous section, we find that the Prokhorov metric is well-suited for small persistence diagrams.
Large scale computations can be improved by the technique of entropic regularization from the theory of optimal transport \cite{lacombe_large_2018}.
As the classical Prokhorov metric admits an optimal transport characterization, our discrete variant might be tractable using similar techniques.

A major aspect of the importance of the Bottleneck distance is its algebraic formulation in terms of interleavings.
This theory generalizes to incorporate the family of Prokhorov metrics.
An algebraic formulation would also provide a perspective on generalizations to multiparameter persistence.

Our results in section \ref{subsec:MetricProperties} establish that our construction yields a Polish space.
This makes it suitable for statistical inference.
In a similar vein, one can also investigate bottleneck profiles persistence diagrams arising from random geometric complexes.
What kind of limit objects appear in this context?
Can they be used to perform statistical testing?

Morally, stability theorems should involve related metrics on the input point cloud and on the persistence diagram side. 
This motivates to investigate Prokhorov-type distances for point clouds in $\R^n$.
Such distances might be useful throughout data science.

\subsection*{Code availability}\label{sec:implementation}
We provide an implementation as a part of a custom gudhi fork at \url{https://github.com/nihell/persistence-prokhorov}.
It is a modification of the GUDHI implementation of the Bottleneck distance \cite{gudhi:BottleneckDistance}.
Let us first illustrate how to use it before we come to runtime considerations.
The algorithm is implemented in C++ and comes with Python bindings.

\begin{verbatim}
prokhorov_distance(diagram_1: numpy.ndarray[numpy.float64],
                   diagram_2: numpy.ndarray[numpy.float64],
                   coef: numpy.ndarray[numpy.float64]) -> float  
\end{verbatim}

It asks for three inputs: \verb|diagram_1|, \verb|diagram_2| and \verb|coef|.
The two diagrams need to be presented as 2D numpy arrays.
The third parameter is a 1D numpy array representing the coefficients of a polynomial to be used as $f$.
Note that the zeroth entry needs to be zero in order to obtain a metric, compare Lemma \ref{lem:GenPrDefininite}.
However, setting the polynomial to be a constant integer one recovers the values of $D_{X,Y}$, which is a feature.
In the technical details, our approach follows \cite{gudhi:BottleneckDistance}, which follows \cite{Kerber2017Geometry}.

In addition, we also add the Prokhorov metric to \cite{gudhi:PersistenceRepresentations}, allowing for parallel computations of distance matrices and integration with \verb|sklearn|.

\subsubsection*{Ackknowledgement}
This work was in part supported by he Centre for Topological Data Analysis, EPSRC grant EP/R018472/1, and by the Dioscuri program initiated by the Max Planck Society, jointly managed with the National Science Centre (Poland), and mutually funded by the Polish Ministry of Science and Higher Education and the German Federal Ministry of Education and Research.
We thank Gesine Reinert and Sayan Mukherjee for valuable discussions and Davide Gurnari for providing help with the experiments.
Finally, we thank the anonymous referees for their helpful suggestions.
\printbibliography

\end{document}

%% file: FigMeasureProkhorov.tex
\begin{figure}[htb]
	\centering
	     \begin{tikzpicture}[scale=0.6]
		    \draw (0,0)--(8,0);
		    \draw (0,0)--(0,8);
		    \shade[inner color=red,outer color=white] (2,1.5) rectangle (4,3.5);
		    
		    \draw (9,0)--(17,0);
		    \draw (9,0)--(9,8);
		    \shade[inner color=red,outer color=white, cm={cos(15) ,-sin(15) ,sin(15) ,cos(15) ,(2,3)}] (9.5,1.5) rectangle (11,3);
		    \shade[inner color=red,outer color=white] (16.5,7) circle (0.2);
		 \end{tikzpicture}
    	\caption{Illustration of two Prokhorov-close measures\label{fig:ProkhorovExmp} which are not Wasserstein-close.}
\end{figure}

%% file: FigCoupling.tex
\begin{figure}[htb]
	\centering
	     \begin{tikzpicture}[
    /pgfplots/y=1cm, /pgfplots/x=1cm 
]
		    \begin{axis}[
		    anchor=origin, 
            no markers, samples=100,
            axis lines*=center, 
            xtick=\empty, ytick=\empty,
            xmin=0, ymax=0, clip=false,
            enlargelimits=false, 
            hide y axis
            ]
            \addplot [fill=red, draw=none, domain=0:8, yshift = -0.1] {-1/(0.5*sqrt(2*pi))*exp(-((x-5.5)^2)/(2*0.5^2))} node[above,pos=0.68] {$\mu$}\closedcycle ;
        \end{axis}
        \begin{axis}[
            anchor=origin, 
            rotate around={90:(current axis.origin)}, 
            no markers, samples=1000,
            axis lines*=left, 
            xtick=\empty, ytick=\empty,
            xmin=0, ymin=0, xmax = 2,
            enlargelimits=false, clip=false, 
            hide y axis 
            ]
            \addplot [fill=blue, draw=none, domain=0:8, yshift = 0.1] {1/(0.1*sqrt(2*pi))*exp(-((x-1)^2)/(2*0.1^2)))} node[above,pos=0.5] {$\nu$}\closedcycle;
        \end{axis}
		    \shade[inner color=violet,outer color=white] (5.4,0.9) ellipse (2 and 0.5);
		    \node at (5.4,0.9){$\gamma$};

		 \end{tikzpicture}
    	\caption{Illustration of two measures $\mu$ and  $\nu$ and a coupling $\gamma$ of them.\label{fig:CouplingExmp}}
\end{figure}

%% file: references.bib
@article{atienzaStabilityPersistentEntropy2020a,
  title = {On the Stability of Persistent Entropy and New Summary Functions for Topological Data Analysis},
  author = {Atienza, Nieves and {Gonzalez-D{\'i}az}, Rocio and {Soriano-Trigueros}, Manuel},
  year = {2020},
  month = nov,
  journal = {Pattern Recognition},
  volume = {107},
  pages = {107509},
  issn = {0031-3203},
  doi = {10.1016/j.patcog.2020.107509},
  langid = {english}
}

@book{grafakosClassicalFourierAnalysis2014,
  title = {Classical {{Fourier Analysis}}},
  author = {Grafakos, Loukas},
  year = {2014},
  series = {Graduate {{Texts}} in {{Mathematics}}},
  volume = {249},
  publisher = {{Springer New York}},
  address = {{New York, NY}},
  doi = {10.1007/978-1-4939-1194-3},
  isbn = {978-1-4939-1193-6 978-1-4939-1194-3},
  langid = {english}
}

@article{gibbs,
	  title={On choosing and bounding probability metrics},
  author={Gibbs, Alison L and Su, Francis Edward},
  journal={International statistical review},
  volume={70},
  number={3},
  pages={419--435},
  year={2002},
  publisher={Wiley Online Library}
}

@misc{EdelsbrunnerHarer2010CompTop,
 Author = {Herbert {Edelsbrunner} and John L. {Harer}},
 Title = {Computational topology. An introduction.},
 ISBN = {978-0-8218-4925-5/hbk},
 Pages = {xii + 241},
 Year = {2010},
 Publisher = {Providence, RI: American Mathematical Society (AMS)},
 Language = {English},
 MSC2010 = {55-01 55-04 55-02 55T05},
 Zbl = {1193.55001}
}

@Article{Kerber2017Geometry,
 Author = {Michael {Kerber} and Dmitriy {Morozov} and Arnur {Nigmetov}},
 Title = {{Geometry helps to compare persistence diagrams.}},
 FJournal = {{ACM Journal of Experimental Algorithmics}},
 Journal = {{ACM J. Exp. Algorithm.}},
 ISSN = {1084-6654/e},
 Volume = {22},
 Pages = {20},
 Note = {Id/No 1.4},
 Year = {2017},
 Publisher = {Association for Computing Machinery (ACM), New York, NY},
 Language = {English},
 MSC2010 = {68U05 55N35 68R10},
 Zbl = {1414.68129}
}

@Article{Efrat2001Geometry,
 Author = {A. {Efrat} and A. {Itai} and M. J. {Katz}},
 Title = {{Geometry helps in bottleneck matching and related problems.}},
 FJournal = {{Algorithmica}},
 Journal = {{Algorithmica}},
 ISSN = {0178-4617; 1432-0541/e},
 Volume = {31},
 Number = {1},
 Pages = {1--28},
 Year = {2001},
 Publisher = {Springer US, New York, NY},
 Language = {English},
 MSC2010 = {68T10 68R10},
 Zbl = {0980.68101}
}

@inproceedings{chazal2009gromov,
  title={Gromov-Hausdorff stable signatures for shapes using persistence},
  author={Chazal, Fr{\'e}d{\'e}ric and Cohen-Steiner, David and Guibas, Leonidas J and M{\'e}moli, Facundo and Oudot, Steve Y},
  booktitle={Computer Graphics Forum},
  volume={28},
  number={5},
  pages={1393--1403},
  year={2009},
  organization={Wiley Online Library}
}

@Article{Crawley-Boevey2015Decomposition,
 Author = {William {Crawley-Boevey}},
 Title = {{Decomposition of pointwise finite-dimensional persistence modules.}},
 FJournal = {{Journal of Algebra and its Applications}},
 Journal = {{J. Algebra Appl.}},
 ISSN = {0219-4988; 1793-6829/e},
 Volume = {14},
 Number = {5},
 Pages = {8},
 Note = {Id/No 1550066},
 Year = {2015},
 Publisher = {World Scientific, Singapore},
 Language = {English},
 MSC2010 = {16G20},
 Zbl = {1345.16015}
}

@Article{Mileyko2011Probability,
 Author = {Yuriy {Mileyko} and Sayan {Mukherjee} and John {Harer}},
 Title = {{Probability measures on the space of persistence diagrams}},
 FJournal = {{Inverse Problems}},
 Journal = {{Inverse Probl.}},
 ISSN = {0266-5611},
 Volume = {27},
 Number = {12},
 Pages = {22},
 Note = {Id/No 124007},
 Year = {2011},
 Publisher = {IOP Publishing, Bristol},
 Language = {English},
 MSC2010 = {68U05 55N35 65D18 68Q87 60D05 60B05},
 Zbl = {1247.68310}
}

@incollection{gudhi:BottleneckDistance
, author =  "Fran{\c{c}}ois Godi"
, title =   "Bottleneck distance"
, publisher =  "{GUDHI Editorial Board}"
, edition =     "{3.4.1}"
, booktitle =   "{GUDHI} User and Reference Manual"
, year =        2021
}

@Article{Adams2017Persistence,
 Author = {Henry {Adams} and Tegan {Emerson} and Michael {Kirby} and Rachel {Neville} and Chris {Peterson} and Patrick {Shipman} and Sofya {Chepushtanova} and Eric {Hanson} and Francis {Motta} and Lori {Ziegelmeier}},
 Title = {{Persistence images: a stable vector representation of persistent homology}},
 FJournal = {{Journal of Machine Learning Research (JMLR)}},
 Journal = {{J. Mach. Learn. Res.}},
 ISSN = {1532-4435; 1533-7928/e},
 Volume = {18},
 Pages = {35},
 Note = {Id/No 8},
 Year = {2017},
 Publisher = {Microtome Publishing, Brookline, MA},
 Language = {English},
 MSC2010 = {68T05 55N31 68T09},
 Zbl = {1431.68105}
}

@book{rachevProbabilityMetricsStability1991,
 Author = {Svetlozar T. {Rachev}},
 Title = {Probability metrics and the stability of stochastic models},
 FJournal = {{Wiley Series in Probability and Mathematical Statistics}},
 Journal = {{Wiley Ser. Probab. Math. Stat.}},
 ISBN = {0-471-92877-1/hbk},
 Pages = {xiv + 494},
 Year = {1991},
 Publisher = {Chichester etc.: John Wiley \& Sons Ltd.},
 Language = {English},
 MSC2010 = {60A10 60-02},
 Zbl = {0744.60004}
}

@article{peyreComputationalOptimalTransport2020,
    title={Computational optimal transport: With applications to data science},
  author={Peyr{\'e}, Gabriel and Cuturi, Marco and others},
  journal={Foundations and Trends{\textregistered} in Machine Learning},
  volume={11},
  number={5-6},
  pages={355--607},
  year={2019},
  publisher={Now Publishers, Inc.}
}

@misc{gromov2007metric,
 Author = {Gromov, Misha},
 Editor = {LaFontaine, J. and Pansu, P.},
 Title = {Metric structures for {Riemannian} and non-{Riemannian} spaces. {Transl}. from the {French} by {Sean} {Michael} {Bates}. {With} appendices by {M}. {Katz}, {P}. {Pansu}, and {S}. {Semmes}. {Edited} by {J}. {LaFontaine} and {P}. {Pansu}},
 Edition = {3rd printing},
 FSeries = {Modern Birkh{\"a}user Classics},
 Series = {Mod. Birkh{\"a}user Classics},
 ISSN = {2197-1803},
 ISBN = {978-0-8176-4582-3},
 Year = {2007},
 Publisher = {Basel: Birkh{\"a}user},
 Language = {English},
 zbMATH = {5114904},
 Zbl = {1113.53001}
}

@book{villaniTopicsOptimalTransportation2003,
  title = {Topics in Optimal Transportation},
  author = {Villani, Cédric},
  date = {2003},
  journaltitle = {Graduate Studies in Mathematics},
  volume = {58},
  publisher = {{American Mathematical Society (AMS)}},
  location = {{Providence, RI}},
  issn = {1065-7338},
  isbn = {978-0-8218-3312-4},
  langid = {english},
  pagetotal = {xvi + 370 p},
  zmnumber = {1106.90001}
}

@article{prokhorovConvergenceRandomProcesses1956,
  title = {Convergence of random processes and limit theorems in probability theory},
  author = {Prokhorov, Yu V.},
  date = {1956},
  journaltitle = {Teor. Veroyatn. Primen.},
  volume = {1},
  pages = {177--238},
  publisher = {{Russian Academy of Sciences - RAS (Rossi\textbackslash u ıskaya Akademiya Nauk - RAN), Moscow; Nauka, Moscow}},
  issn = {0040-361X; 2305-3151/e},
  langid = {russian},
  zmnumber = {0075.29001}
}

@article{bubenikTopologicalSpacesPersistence2018,
  title = {Topological Spaces of Persistence Modules and Their Properties},
  author = {Bubenik, Peter and Vergili, Tane},
  date = {2018-12},
  journaltitle = {J Appl. and Comput. Topology},
  volume = {2},
  pages = {233--269},
  issn = {2367-1726, 2367-1734},
  doi = {10.1007/s41468-018-0022-4},
  archiveprefix = {arXiv},
  eprint = {1802.08117},
  eprinttype = {arxiv},
  langid = {english},
  number = {3-4}
}

@article{blumbergRobustStatisticsHypothesis2014,
  title={Robust statistics, hypothesis testing, and confidence intervals for persistent homology on metric measure spaces},
  author={Blumberg, Andrew J and Gal, Itamar and Mandell, Michael A and Pancia, Matthew},
  journal={Foundations of Computational Mathematics},
  volume={14},
  number={4},
  pages={745--789},
  year={2014},
  publisher={Springer}
}

@misc{carriereSKlearnTDA,
    title = {3D shape segmentation using TDA},
    author = {Mathieu Carrière},
    type = {IPython notebook},
    url = {https://github.com/MathieuCarriere/sklearn-tda/tree/master/example/3DSeg}
}

@incollection{gudhi:PersistenceRepresentations
, author =  "Pawel Dlotko and Mathieu Carrière and Martin Royer"
, title =   "Persistence representations"
, publisher =  "{GUDHI Editorial Board}"
, edition =     "{3.4.1}"
, booktitle =   "{GUDHI} User and Reference Manual"
, year =        2021
}

@article{Chen:2009:ABF,
  author = "Xiaobai Chen and Aleksey Golovinskiy and Thomas Funkhouser",
  title = "A Benchmark for {3D} Mesh Segmentation",
  journal = "ACM Transactions on Graphics (Proc. SIGGRAPH)",
  year = "2009",
  month = aug,
  volume = "28",
  number = "3"
}

@article{hopcroft_n52_1973,
	title = {An $n^{5/2}$ Algorithm for Maximum Matchings in Bipartite Graphs},
	volume = {2},
	issn = {0097-5397},

	doi = {10.1137/0202019},
	pages = {225--231},
	number = {4},
	journaltitle = {{SIAM} J. Comput.},
	author = {Hopcroft, John E. and Karp, Richard M.},

	date = {1973-12-01},
	note = {Publisher: Society for Industrial and Applied Mathematics},
}

@article{scikit-learn,
 title={Scikit-learn: Machine Learning in {P}ython},
 author={Pedregosa, F. and Varoquaux, G. and Gramfort, A. and Michel, V.
         and Thirion, B. and Grisel, O. and Blondel, M. and Prettenhofer, P.
         and Weiss, R. and Dubourg, V. and Vanderplas, J. and Passos, A. and
         Cournapeau, D. and Brucher, M. and Perrot, M. and Duchesnay, E.},
 journal={Journal of Machine Learning Research},
 volume={12},
 pages={2825--2830},
 year={2011}
}

@article{schubert_fast_2020,
	title = {Fast and Eager k-Medoids Clustering: O(k) Runtime Improvement of the {PAM}, {CLARA}, and {CLARANS} Algorithms},

	shorttitle = {Fast and Eager k-Medoids Clustering},
	journaltitle = {{arXiv}:2008.05171 [cs, stat]},
	author = {Schubert, Erich and Rousseeuw, Peter J.},

	date = {2020-08-12},
	langid = {english},
	eprinttype = {arxiv},
	eprint = {2008.05171},
}

@misc{schubert_erichschuberttu-dortmundde_kmedoids_nodate,
	title = {kmedoids 0.1.5-dev       : k-Medoids clustering with the {FasterPAM} algorithm},
	rights = {{GNU} General Public License v3 or later},
	url = {https://github.com/kno10/python-kmedoids},
	shorttitle = {kmedoids 0.1.5-dev},
	author = {Schubert, Erich},
	urldate = {2021-05-19},
}

@article{lacombe_large_2018,
	title = {Large Scale computation of Means and Clusters for Persistence Diagrams using Optimal Transport},
	journaltitle = {{arXiv}:1805.08331 [cs, stat]},
	author = {Lacombe, Théo and Cuturi, Marco and Oudot, Steve},
	date = {2018-11-13},
	langid = {english},
	eprinttype = {arxiv},
	eprint = {1805.08331},
}

@article{chacholskiMetricsStabilizationOne2020,
    title={Metrics and stabilization in one parameter persistence},
  author={Chach{\'o}lski, Wojciech and Riihimaki, Henri},
  journal={SIAM Journal on Applied Algebra and Geometry},
  volume={4},
  number={1},
  pages={69--98},
  year={2020},
  publisher={SIAM}
}
